\documentclass[onecolumn,draftcls]{IEEEtran}
\def\comment#1{}
\usepackage{graphicx,amsfonts,amsmath,amssymb,xcolor,amsthm,bbm,mathrsfs}
\usepackage{algorithm,algorithmic,bm,color}
\usepackage{multirow}
\usepackage[noadjust]{cite}
\usepackage{adjustbox}
\usepackage{times, graphics, color,  dsfont,setspace}
\def\comment#1{}
\newcommand{\argmin}{\arg\!\min}

\newcommand{\Dmat}{{\bf D}}

\newcommand{\Hmat}[0]{{{\bf H}}}
\newcommand{\Imat}{{\bf I}}

\newcommand{\Rmat}[0]{{{\bf R}}}

\newcommand{\Xmat}{{\bf X}}
\newcommand{\Ymat}[0]{{{\bf Y}}}

\newcommand{\cv}{{\boldsymbol{c}}}

\newcommand{\ev}[0]{{\boldsymbol{e}}}

\newcommand{\sv}[0]{{\boldsymbol{s}}}

\newcommand{\wv}{\boldsymbol{w}}

\newcommand{\xv}{\boldsymbol{x}}
\newcommand{\yv}{\boldsymbol{y}}

\newcommand{\zv}{\boldsymbol{z}}

\newcommand{\Phimat}{\boldsymbol{\Phi}}

\newcommand{\thetav}{\boldsymbol{\theta}}

\newcommand{\kappav}{{\boldsymbol{\kappa}}}

\newcommand{\xiv}{{\boldsymbol{\xi}}}

\newcommand{\ts}{^{\top}}
\newcommand{\inv}{^{-1}}
\newcommand{\ie}{{\em i.e.}}
\newcommand{\eg}{{\em e.g.}~}

\newtheorem{definition}{Definition}
\newtheorem{theorem}{Theorem}
\newtheorem{corollary}{Corollary}

\newtheorem{lemma}{Lemma}
\newtheorem{remark}{Remark}
\newtheorem{assumption}{Assumption}

\newcommand{\diag}{\mathrm{diag}}

\def\setR{\mathbbm{R}}

\def\F{\mathcal{F}}

\def\E{\mathcal{E}}

\def\C{\mathcal{C}}
\def\N{\mathcal{N}}

\newcommand{\LKr}[1]{\Big\{#1\Big\}}
\newcommand{\LPr}[1]{\left(#1\right)}
\newcommand{\LCr}[1]{\left[#1\right]}
\newcommand{\LPd}[1]{\left<#1\right>}
\newcommand{\Labs}[1]{\Big|#1\Big|}

\newcommand{\Eox}[1]{{\rm E}\left[#1\right]}

\newcommand{\mvec}[1]{{\boldsymbol #1}}

\newcommand{\gos}{\rightarrow}

\newcommand{\Lp}[3]{{\left\|{#1}\right\|}_{#2}^{#3}}

\begin{document}
\setlength{\parskip}{.02in}
	%\setlength{\parskip}{.02in}
	
	%	\title{Maximum A Posteriori Estimation and Euclidean Projection for Compressive Sensing}
	
\title{Snapshot compressed sensing: performance bounds and algorithms}

\author{Shirin Jalali and Xin Yuan
		\thanks{The authors are with Nokia Bell Labs, 600 Mountain Avenue, Murray Hill, NJ, 07974, USA, \{shirin.jalali, xin\_x.yuan\}@nokia-bell-labs.com}
		\footnote{This paper was  presented in part at 2018 IEEE International Symposium on Information Theory, Vail, Colorado \cite{Jalali18ISIT}.}
}
\maketitle

\begin{abstract}
Snapshot compressed  sensing (CS) refers to compressive imaging systems in which  multiple frames are mapped into a single measurement frame. Each pixel in the  acquired frame is  a noisy linear mapping of the corresponding pixels in the  frames that are combined  together.  While the problem can be cast as a  CS problem,  due to the very special  structure of the sensing  matrix,  standard CS theory cannot be employed to study such systems.  In this paper, a compression-based framework is employed for theoretical analysis of snapshot CS systems. It is shown that this  framework leads to two novel, computationally-efficient and theoretically-analyzable compression-based  recovery algorithms.  The proposed methods are  iterative and  employ  compression codes to define and impose the structure of the desired signal. 
Theoretical convergence guarantees are derived for both algorithms. In the simulations, it is shown that,  in the cases of both  noise-free  and noisy measurements,  combining the proposed algorithms with a customized video compression code, designed to exploit nonlocal structures  of video frames, significantly improves the state-of-the-art performance.
\end{abstract}
	
%\begin{IEEEkeywords}
%	Compressive sensing, computational imaging, video processing, information theory, compression, structure.
%\end{IEEEkeywords}
%	

\section{Introduction}\label{sec:intro}
\subsection{Problem statement}
The problem of compressed sensing (CS), recovering high-dimensional vector $\xv\in\mathds{R}^n$ from its noisy under-determined linear measurements $\yv = \Phimat \xv + \zv$, where $\yv,\zv\in\mathds{R}^m$ and  $\Phimat\in\mathds{R}^{m\times n}$, has been the subject of various theoretical and algorithmic  studies in the past decade. Clearly, since such systems of linear equations are underdetermined, the recovery of $\xv$ from measurements $\yv$ is only feasible, if the input signal is {\em structured}. For various types of structure, such as sparsity, group-sparsity, etc., it is known that efficient algorithms exist that efficiently and robustly recover  $\xv$ from measurements $\yv$.  Starting by  the seminal works of \cite{CaRoTa06} and \cite{Donoho06ITT}, there have been significant theoretical advances in  this area. Initially, most   such  theoretical results were developed assuming that the entries of the sensing matrix $\Phimat$  are  independently and identically distributed (i.i.d.) according to   some  distribution. In the meantime, various modern  {\em compressive imaging} systems have been built during the last decade or so \cite{Duarte08SPM,Gao14_Nature,Gehm07,Hitomi11ICCV,Patrick13OE,Reddy11CVPR} that are based on solving ill-posed linear inverse problems. 
%Sensing matrices employed by many of such imaging systems  are structured in forms that do not comply with standard dense random  sensing matrices. 
Convincing results have  been obtained in diverse applications, such as  video CS~\cite{Hitomi11ICCV,Patrick13OE,Reddy11CVPR} and hyper-spectral image CS~\cite{Gehm07,Wagadarikar08CASSI,Wagadarikar09CASSI}. However, except for the single-pixel camera~\cite{Duarte08SPM} and similar architectures~\cite{Huang13ICIP,Aswin15SIAM}, the sensing matrices employed in most of these practical  systems are usually not random, and typically {\em structured}  very differently compared to dense  random sensing matrices  studied in the CS literature. Therefore, more recently, there has also been significant effort on analyzing CS systems that employ  structured sensing matrices. (Refer to \cite{Bajwa07Toeplitz,Romberg09Conv,candes2011probabilistic, Eftekhari15RIP_block,foucart2017mathematical, adcock2017breaking,boyer2017compressed} for some examples   of such results.) 

One important example of practical imaging systems built upon CS ideas is a  hyperspectral compressive imaging system  called coded aperture snapshot spectral imaging (CASSI) ~\cite{Gehm07}. CASSI   recovers  a three-dimensional (3D) spectral data cube, in which more than 30 frequency channels (images) at different wavelengths have been reconstructed, from a single two-dimensional (2D) captured measurement. This coded aperture modulating strategy has paved the way for many high-dimensional compressive imaging systems, from the aforementioned video CS to depth CS~\cite{Llull15Optica,Yuan14CVPR}, polarization CS~\cite{Tsai15OE}, and joint temporal-spectral CS~\cite{Tsai15OL}. 

The measurement process in such hardware systems, known as snapshot CS systems,   can typically  be  modeled as~\cite{Patrick13OE,Wagadarikar09CASSI}%\footnote{Specifically, in optics, researchers usually use the model $\gv = \Hmat\fv$, where $\{\gv,\fv\}$ are the measurement and signal, respectively. In this paper, following the terminology in the CS literature, we still use $\{\yv,\xv\}$ to denote the measurement and signal.}
\begin{equation} \label{Eq:ghf}
{\yv = \Hmat \xv + \zv},
\end{equation}
where, $\xv\in\mathds{R}^{nB}$,  and  $\yv\in\mathds{R}^{n}$ denote the desired signal, and the measurement vector, respectively. Here, $B$ denotes the number of $n$-dimensional input vectors (frames) that are combined together. In other words, the input $\xv$ is a multi-frame signal that consists of $B$ $n$-dimensional vectors as 
\begin{equation} \label{Eq:xv1toB}
%\xv = \left[
%\xv_1\ts,\dots,\xv_B\ts
%\right]\ts,
\xv = \left[\begin{array}{c}
\xv_1\\
\vdots\\
\xv_B
\end{array}\right],
\end{equation}
where $\xv_i\in\mathds{R}^n$, $i=1,\ldots,B$.
 The measured signal is an $n$-dimensional vector. Thus,  the sampling rate of this  system   is equal to   ${n\over nB}={1\over B}$. The main difference between a standard CS system and a snapshot CS system lies  in their sensing matrices. The sensing matrix  $\Hmat$ used in a snapshot CS   system follows a very specific structure and can be written as
\begin{equation} \label{Eq:Hmat_strucutre}
{\Hmat = \left[\Dmat_1,\dots, \Dmat_B\right]},
\end{equation}
where  $\Dmat_k\in\mathds{R}^{n\times n}$, $k=1,\ldots,B$, are {\em diagonal} matrices. It can be observed that, unlike dense matrices used in standard CS, here, the sensing matrix    is very sparse. That is, of the $n^2B$ entries in matrix $\Hmat$, at most, $nB$ of them are non-zero.

In this paper, we focus on and analyze such snapshot compressed sensing systems. As discussed before, a common property of such   acquisition  systems is that, due to some  hardware constraints, each measurement only depends on few entries  of the input signal.  Furthermore, the location of the non-zero elements of the measurement kernel follows a very specific pattern. For example, in some video CS systems, high-speed frames are modulated at a higher frequency than the capture rate of the camera, which is working at a low frame rate. Each measured pixel  in the captured frame is a function of  the pixels located at the same position in the input frames. 
In this manner, each captured measurement frame can recover a number of high-speed frames, depending on the coding strategy, e.g., 148 frames are reconstructed from a snapshot measurement in~\cite{Patrick13OE}. 
%Due to their special structures, it is for instance  straightforward to check that such matrices do not satisfy standard conditions, such as RIP or  mutual coherence~\cite{Candes06sparsityand} that are used in  theoretical analysis of standard CS systems.\footnote{For instance, for a sensing matrix of the form described in \eqref{Eq:Hmat_strucutre}, $i)$ the {\em spark}~\cite{Gorodnitsky97_focuss} of $\Hmat$ is ${\rm spark}(\Hmat) = 2$, since for any $j=1,\dots,n-1$, $\hv_{j}$ and $\hv_{j+B}$ are linearly dependent; recall that the {\em spark} of a given matrix is the smallest number of columns that are linearly-dependent, and  $ii)$ the mutual coherence~\cite{Donoho02optimallysparse} of $\Hmat$ is $\mu(\Hmat) = 1$. }  Hence, snapshot CS systems have resisted theoretical analysis, despite the fact that they have delivered promising results in practice.
%Such sensing matrices have not been studied theoretically in the past. 

 In this paper, we provide the  first theoretical analysis of snapshot CS systems and also propose new theoretically-analyzable robust recovery algorithms that achieve the state-of-the-art performance.  As mentioned before, the main theoretical challenge is the fact that the sensing matrix is sparse and follows a very special structure.

\subsection{Contributions of this paper}

As discussed in the last section, the main goal of our paper is to provide  theoretical analysis for the snapshot CS system. More specifically, we aim to address the following fundamental questions regarding such systems:
\begin{itemize}%{\labelitemi}{\leftmargin=14pt \topsep=-2pt \parsep=0pt}
	\item [1)] Is it {\em theoretically} possible to recover $\xv$ from the measurement $\yv$ defined in \eqref{Eq:ghf}, for $B>1$?
	\item [2)] What is the maximum number of frames $B$ that can be mapped to a single measurement frame and still be  recoverable, and how is this number  related  to the  properties of the signal?
	\item [3)] Are there  efficient  and {\em theoretically-analyzable}  recovery  algorithms for  snapshot CS?
\end{itemize}
%As a byproduct of answering these questions, we obtain two recovery algorithms, called  ``compression-based PGD (CbPGD)'' and ``compression-based GAP (CbGAP)'' , that exhibits the state-of-the-art performance for snapshot CS systems, in our simulation results.  
Inspired by the idea of  compression-based CS~\cite{jalali2016compression}, we develop  a theoretical framework for snapshot CS. We also propose two efficient iterative  snapshot CS  algorithms, called  ``compression-based PGD (CbPGD)'' and ``compression-based GAP (CbGAP)'',   with convergence guarantees. The algorithms achieve state-of-the-art performance in snapshot video CS in our simulations.   Though various algorithms, \eg,~\cite{Yang14GMM,Yang14GMMonline,Yuan15JSTSP}, have been developed for video and hyper-spectral image CS, to our best knowledge, no theoretical guarantees have been  available yet  for the special structure of sensing matrices that arise in snapshot CS. 

%Recently, the first two questions were partially addressed in \cite{Anonymous18ISIT}, where the authors proposed a new  compression-based framework for theoretical analysis of such systems. The results of  \cite{Anonymous18ISIT} answer the first question affirmatively, and show that if the number of frames $B$ is smaller than some constant depending on the signal structure, the proposed compression-based recovery method is able to recover the $B$ frames from a single measured frame. However, the proposed recovery method in \cite{Anonymous18ISIT} is based on an optimization which is not feasible in practice. (We review the results later in Section \ref{sec:CSP}.)  In this paper, inspired by the results of \cite{Anonymous18ISIT}, we address the third question. More precisely, we propose two novel compression-based  snapshot CS algorithms that are theoretically analyzable. We test the performance of the proposed algorithms in video snapshot CS and  observe that the algorithms  achieve state-of-the-art performance.   Though various algorithms, \eg,~\cite{Hitomi11ICCV,Reddy11CVPR,Tan16JSTSP,Yang14GMM,Yuan15JSTSP}, have been developed for video and hyper-spectral image CS, to our best knowledge, no theoretical guarantees have been  available yet  for the special structure of sensing matrices that arise in snapshot CS.

\subsection{Related work}
As mentioned earlier, theoretical work in the CS literature is mainly focused on sparse signals \cite{Candes06sparsityand,Donoho06ITT}  and their extensions such as group-sparsity \cite{Huang09LSS}, model-based sparsity \cite{Baraniuk10ITT}, and the low-rank property~\cite{Dong14TIP}.
Many classes of  signals such as natural images and videos typically follow much more complex patterns than these structures. A recovery algorithm that takes advantage of those complex structures,   potentially, can outperform standard schemes by requiring a lower sampling rate or having a better reconstruction quality. However, designing and analyzing such recovery algorithms that impose both the measurement constraints and the source's known patterns is in general very challenging. One recent approach to address this issue is to take advantage of algorithms that are designed for other data processing tasks such as denoising or data compression and to derive for instance denoising-based \cite{DAMP} or compression-based \cite{jalali2016compression} recovery algorithms. The advantage of this approach  is that,  without much additional effort, it elevates the  scope of structures used by  CS recovery algorithms   to those used by denoising or compression algorithms. As an example, modern image compression algorithms such as JPEG and JPEG2000~\cite{JPEGbook} are very efficient codes that are designed to exploit various common properties of natural images. Therefore, a CS recovery algorithm that employs  JPEG2000 to impose structure on the recovered signal, ideally is a recovery algorithm that searches for a signal that is consistent with the measurements and at the same time satisfies the image properties used by the JPEG2000 code.  This line of work of designing compression-based CS was first started in \cite{jalali2016compression} and then later continued in \cite{beygi2017efficient}, where the authors proposed an efficient compression-based recovery algorithm that achieves state-of-the-art performance in image CS. 

Additionally, there are other CS systems such as those studied in \cite{Holloway12ICCP,Aswin15SIAM} for video CS, and  \cite{Arguello13AO} for hyperspectral imaging. (Please refer to the survey in~\cite{Baraniuk17SPM,Arce14SPM,Cao16SPM} for more details on such systems.)
While the sensing matrices used  in   these systems are not exactly as the sensing matrix defined in \eqref{Eq:Hmat_strucutre}, they  have key   similarities, as in both cases they   are very sparse in a very structured manner. Therefore, we expect  that our proposed compression-based framework, with some moderate modifications, to be applicable to  such cases as well and to pave the way for performing theoretical analysis of such systems too.  
%\subsection{Paper organization}
%The rest of this paper is organized as follows. Section \ref{Sec:SCS} briefly reviews the definitions of lossy compression codes and the performance  of a compression-based snapshot CS recovery method proposed in \cite{Jalali18ISIT}. Section \ref{Sec:PGD} introduces a compression-based recovery optimization for snapshot CS and derives the convergence guarantee.  Section~\ref{Sec:GAP} introduces a fixed step-size algorithm for snapshot CS reconstruction with fast convergence. Simulation results of video CS are shown in Section~\ref{Sec:res}.  
%Section \ref{Sec:con} concludes the paper. %The main steps of the proofs  are presented in Appendix~\ref{Sec:proof}. 
%The proofs of theorems are given in  the supplementary material (SM). 

Finally, as mentioned in the introduction, the focus of this paper is on the imaging systems that can be approximated as a snapshot CS system. There are other types of imaging systems and CS systems with other types of constraints on the sensing matrices \cite{Bajwa07Toeplitz,Romberg09Conv,candes2011probabilistic, Eftekhari15RIP_block,foucart2017mathematical, adcock2017breaking,boyer2017compressed}. However,  the  sensing matrices considered therein are different from the one  used  in snapshot CS systems and hence those results are not applicable to such systems. 

\subsection{Notation}

Matrices are denoted by upper-case bold letters such as $\Xmat$ and $\Ymat$.   Vectors are denoted by bold lower-case letters, such as $\xv$ and $\yv$. For $\xv\in\mathds{R}^n$ and $\yv\in\mathds{R}^n$, $\langle \xv,\yv\rangle=\sum_{i=1}^nx_iy_i$ denotes their inner product.
Sets are denoted by calligraphic letters such as ${\cal X}$ and ${\cal Y}$. The size of a set ${\cal X}$ is denoted as $|\cal X|$. 
Throughout the paper, $\log$ and $\ln$ refer to  logarithm in base 2 and natural logarithm, respectively. 
%For a vector $\xv=(x_1,\ldots,x_n)\in\setR^n$, $(i,j)$, $1\leq i \leq j\leq n$, $x_i^j=(x_i,x_{i+1},\ldots,x_j)$. For $i=1$, $x_1^j$ is denoted as $x^j$. 

\subsection{Paper organization}
The rest of this paper is organized as follows. Section \ref{Sec:SCS} briefly first reviews lossy compression codes for multi-frame signals  and then develops and analyzes a  compression-based snapshot CS recovery method. Section \ref{Sec:cbSCS}  introduces two different efficient  compression-based recovery methods  for snapshot CS,  in subsections \ref{Sec:PGD} and \ref{Sec:GAP} and proves that they both converge.  Simulation results of video CS are shown in Section~\ref{Sec:simulations}. Section \ref{Sec:proofs} provides  proofs of the main results of the paper and Section~\ref{sec:end}  concludes the paper.

%~~~~~~~~~~~~~~~~~~~~~~~~%~~~~~~~~~~~~~~~~~~~~~~~~%~~~~~~~~~~~~~~~~~~~~~~~~%~~~~~~~~~~~~~~~~~~~~~~~~
%~~~~~~~~~~~~~~~~~~~~~~~~%~~~~~~~~~~~~~~~~~~~~~~~~%~~~~~~~~~~~~~~~~~~~~~~~~%~~~~~~~~~~~~~~~~~~~~~~~~
\section{Data compression for snapshot CS} \label{Sec:SCS}
Our proposed framework studies snapshot CS systems via utilizing data compression codes. In the following, we first briefly review the definitions of lossy compression codes for multi-frame signals, and then develop our snapshot CS theory based on data compression.
\subsection{Data Compression}
Consider a compact set ${\cal Q}\subset\mathbb{R}^{n  B}$.  Each signal $\mvec{x}\in{\cal Q}$,
consists of $B$ vectors (frames) $\{\mvec{x}_1,\ldots,\mvec{x}_B\}$ in $\mathbb{R}^n$.  A lossy compression code of rate $r$, $r\in\mathds{R}^+$ and  can be larger than one, for ${\cal Q}$  is characterized by its encoding mapping $f$, where
\begin{align}
{ f:{\cal Q} \rightarrow \{1,2,\ldots, 2^{nBr}\},} \label{eq:f-def}
\end{align}
and its decoding mapping $g$, where
\begin{align}
{ g:\{1,2,...,2^{nBr}\}\gos \mathbb{R}^{nB}.}\label{eq:g-def}
\end{align}
The {\emph {average  distortion}} between  $\mvec{x}$ and its reconstruction  $\hat{\mvec{x}}$ is defined as 
\begin{equation}
{ d(\mvec{x},\hat{\mvec{x}})\triangleq {1\over nB} \sum_{i=1}^B \Lp{\mvec{x}_i-\hat{\mvec{x}}_i}{2}{2} =  {1\over nB}\Lp{\mvec{x}-\hat{\mvec{x}}}{2}{2}\;},\label{eq:def-delta-B}
\end{equation}
where $\mvec{x}$ is defined in~(\ref{Eq:xv1toB}).
Let $\tilde{\xv} = g(f(\mvec{x})).$
The  distortion  of code $(f,g)$ is denoted by $\delta$, which is defined as the supremum of all achievable average per-frame distortions. That is,
\begin{equation}
{ \delta\triangleq \sup_{\mvec{x}\in\mathcal{Q}} d(\mvec{x},\tilde{\xv}) = \sup_{\mvec{x}\in\mathcal{Q}} {1\over nB}\Lp{\mvec{x}-\tilde{\xv}}{2}{2}.}
\end{equation}
Let  ${\cal C}$ denote the codebook of this code defined as
\begin{equation}
{ {\cal C}=\{g(f(\mvec{x})):\; \mvec{x}\in\cal{Q}\}.}  \label{eq:codebook}
\end{equation}
Clearly, since the code is of rate $r$, $|{\cal C}|\leq 2^{nBr}$. 
Consider  a family of compression code $\{(f_r,g_r)\}_r$ for set  ${\cal Q}\subset\mathbb{R}^{nB}$,  indexed by their rate $r$. The deterministic distortion-rate function of this family of codes is defined as 
\begin{equation}
\delta(r)=\sup_{\xv\in{\cal Q}} {1\over nB}\|\xv-g_r(f_r(\xv))\|_2^2.
\end{equation}
The corresponding deterministic rate-distortion function of this family of codes is defined as
%%\vspace{-3mm}
\[
{r(\delta)=\inf\{r: \delta(r)\leq \delta\}}.
\]
The $\alpha$-dimension of this family of codes is defined as \cite{jalali2016compression}
%%\vspace{-3mm}
\begin{equation}
{\alpha =\limsup_{\delta\to 0} {2r(\delta)\over \log{1\over \delta}}}.\label{eq:alpha-dim}
\end{equation}
It can be shown that in standard CS,  the $\alpha$-dimension of a compression code  is connected to the sampling rate required for a compression-based recovery method that employs this family of codes  to, asymptotically, recover the input at  zero distortion \cite{jalali2016compression}.  

As an example, the set ${\cal Q}$ could represent the set of all $B$-frame natural videos and any video compression code, such as MPEG compression, could play the role of the compression code $(f,g)$. 

In our theoretical derivations, to simplify the proofs, we make an additional assumption about the compression code as follows. 
\begin{assumption}\label{assumption:1}
	In our later theoretical derivations we assume that the compression code is such that $g(f(\xv))$ returns  the codeword in $\cal C$ that is closest to $\xv$. That is, $g(f(\xv))=\argmin_{\cv\in{\cal C}}\|\xv-\cv\|_2^2$. 
\end{assumption}
The above assumption is not  critical in our proofs and in fact it is straightforward to verify that  relaxing this assumption  only affects the reconstruction error by  an additional term that is proportional to how well the mapping $g(f(\cdot))$ approximates the desired projection.

\subsection{Compression-based Recovery}\label{sec:CSP}
%Hereby, we propose a compressible signal pursuit (CSP)-type optimization as a compression-based recovery algorithm for snapshot CS systems. 

While the main body of research in CS has focused on structures such as sparsity and its generalizations, recently, there has been a growing body work that consider much more  general structures. Given the fact that most signals of interest follow structures beyond sparsity, such new schemes potentially are  more efficient in terms of their required sampling  rates or reconstruction quality. 

One approach to develop recovery algorithms that employ more complex structures is to take advantage of already existing  data compression codes. For some classes of signals such as images and videos, after  decades of research,  there exist  efficient compression codes that take advantages of  complex structures.   Compressible signal pursuit (CSP), proposed in \cite{jalali2016compression}, is a compression-based recovery optimization. \cite{jalali2016compression} shows that compression-based CS recovery is possible and  can achieve the optimal performance in terms of  required sampling rates. 

Inspired by the CSP optimization, we propose a CSP-type optimization as a compression-based recovery algorithm for snapshot  measurement systems.  
Consider the compact set ${\cal Q}\subset\mathbb{R}^{nB}$ equipped with a rate-$r$ compression code described by mappings  $(f,g)$, defined in \eqref{eq:f-def}-\eqref{eq:g-def}. 
Consider $\mvec{x}\in\cal{Q}$ and  its snapshot measurement
\begin{equation}\label{Eq:Hx+z}
{\mvec{y}=\Hmat\xv + \zv = \sum_{i=1}^B\Dmat_i\mvec{x}_i + \zv,}
\end{equation} where $\Hmat$ is defined in \eqref{Eq:Hmat_strucutre} and  $\Dmat_i=\diag(D_{i1},\ldots,D_{in})$. Then, a CSP-type recovery, given $\yv$ and  $(\Dmat_1,\ldots, \Dmat_k)$, estimates $\xv$ by solving the following optimization:
\begin{align}
{\hat{\mvec{x}}=\arg\min_{\mvec{c}\in{\cal C}}\|\mvec{y}-\sum_{i=1}^B\Dmat_i\mvec{c}_i\|_2^2,}\label{eq:CSP-block}
\end{align}
where ${\cal C}$ is defined in~\eqref{eq:codebook} and each codeword $\cv$ is broken into $B$ $n$-dimensional blocks $\cv_1,\ldots,\cv_B$ as \eqref{Eq:xv1toB}. In other words, given a measurement vector $\yv$, this optimization, among all compressible signals, i.e., signals in the codebook, picks the one that is closest to the observed measurements.  As mentioned earlier, a key advantage of compression-based recovery methods such as \eqref{eq:CSP-block} is that, without much additional effort, through the use of proper compression codes, they can take advantages of both temporal (spectral) and spatial dependencies that exist in  multi-frame signals, such as videos. 
%It is worth noting that when $$
At $B=1$, with  a {\em traditional} dense  sensing matrix, this reduces to the standard CSP optimization~\cite{jalali2016compression}. However, theoretically, the two setups are significantly  different, and for $B>1$, the original proof of the  CSP optimization does not work in  the snapshot CS setting. 

The following theorem characterizes the performance of this CSP-type recovery method by connecting the parameters of the code, its rate and its distortion, to the number of frames $B$ and the reconstruction quality.

\begin{theorem}\label{thm:main-csp}
	Assume that for all $ \mvec{x}\in{\cal Q}$, $\Lp{\mvec{x}}{\infty}{}\leq {\rho\over 2}$. Further consider a rate-$r$ compression code characterized by codebook $\cal C$ that achieves distortion $\delta$ on $\cal{Q}$.  
	Moreover,  $\Dmat_1,\ldots,\Dmat_B$ are i.i.d., such that, for $i=1,\ldots,B$, $\Dmat_i=\diag(D_{i1},\ldots,D_{in})$, and $\{D_{ij}\}_{j=1}^{n}\stackrel{\rm i.i.d.}{\sim} {\cal N}(0,1)$. For $\xv\in{\cal Q}$ and $\yv=\sum_{i=1}^B\Dmat_i\mvec{x}_i$,  let $\hat{\xv}$ denote the solution of \eqref{eq:CSP-block}.  Let $K=8/3$. Choose $\epsilon>0$, a free parameter, such that    $\epsilon \leq 2K$.  Then,
	\begin{align}
	{ {1\over nB}\|\mvec{x}-\hat{\mvec{x}}\|_2^2\leq \delta +\rho^2\epsilon},
	\end{align}
	with a probability larger than $	 { 1- 2^{nBr+1} {\rm e}^{- {\epsilon^2 n \over 16 K^2 }}}$. 
\end{theorem}

The proof is presented in Section~\ref{append:main_csp}.

%The proofs of Theorem~\ref{thm:main-csp} and Corollary~\ref{Coro:B} below are reproduced in the supplementary material (SM).
\begin{corollary} \label{Coro:B}
	Consider the same setup as in  Theorem \ref{thm:main-csp}. Given $\eta>0$, assume that
		\begin{align}
		 {B<{1\over \eta} ({\log{1\over \delta}\over 2r}).}\label{Eq:col_B}
		\end{align}
	Then, 
	%\vspace{-3mm}
	\begin{align}
	{ \Pr\left({1\over nB}\|\mvec{x}-\hat{\mvec{x}}\|_2^2>\delta +{8\rho^2} \sqrt{\log{1\over \delta} \over \eta}\;\right)\leq 2 {\rm e}^{-  { \log {1\over \delta} \over 5\eta} n}}.  \label{Eq:col_P}
	\end{align}
\end{corollary}
\begin{proof}
	%\vspace{-3mm}
	In Theorem \ref{thm:main-csp}, let 
		\[
		\epsilon ={8} \sqrt{\log{1\over \delta} \over \eta}.
		\]
	Then, according to Theorem \ref{thm:main-csp}, the probability of the  error event  can be upper bounded by
	%%\vspace{-3mm}
	\begin{align}
	2^{nBr+1} {\rm e}^{- {\epsilon^2 n \over 16K^2 }} &\leq 2 {\rm e}^{-  {n\over \eta} \log {1\over \delta}\LPr{{4\over K^2}-{\ln2\over 2} } },
	\end{align}
	where $K=8/3$. But ${4\over K^2}-{\ln2\over 2}>{1\over 5}$. Therefore, the desired result follows. 	
\end{proof}
%
%\vspace{-3mm}

Consider  a family of compression codes $\{(f_r,g_r)\}_r$ for set  ${\cal Q}\subset\mathbb{R}^{nB}$,  indexed by their rate $r$. Roughly speaking, Corollary \ref{Coro:B} states that, as  $\delta\to0$, if $B$ is smaller that $1\over \eta \alpha$, where $\alpha$ denotes the $\alpha$-dimension defined in \eqref{eq:alpha-dim}, the achieved distortion by the CSP-type recovery  is bounded by a constant that is inversely proportional to $1\over \sqrt{\eta}$. (Here, $\eta$ is a free parameter.)
In snapshot CS, as mentioned earlier, the sampling rate is ${1\over B}$. Hence, in other words, to bound the achieved distortion, this corollary requires the sampling rate to exceed  $\eta \alpha$.

To better understand the $\alpha$-dimension of structured multi-frame signals, inspired by video signals, consider the following set of $B$-frame signals in $\mathbb{R}^{nB}$. Assume that the first frame of each $B$-frame signal $\xv\in{\cal Q}$  is an image that has a $k$-sparse representation.  Further assume that the $\ell_2$-norm of $\xv$ is bounded by $1$, i.e., $\|\xv_1\|_2\leq 1$. Also assume that the next $(B-1)$ frames all share the same non-zero entries as $\xv_1$ located arbitrarily across each frame. This very simple model is inspired by  video frames and how consecutive frames are built by shifting the positions of the objects that are in the previous frames. Consider the following simple compression code for signals in $\cal Q$. For the first frame, we first use the orthonormal basis to transform the signal and then describe the locations of the (at most $k$) non-zero entries and their  quantized values, each quantized into a fixed number of bits. Since, by our assumption,  all frames share the same non-zero values, a code for all frames can be built by just coding the locations of the non-zero entries of the remaining $(B-1)$ frames. Changing the number of bits used to quantize each non-zero element yields a family of compression codes operating at different rates and distortions. Since the sparsifying basis is assumed be orthonormal, the $\alpha$-dimension of the code developed for the first frame $\xv_1$ can be shown to be equal to ${k\over n}$ \cite{jalali2016compression}. For the  code developed for $B$-frame signals in $\cal Q$, since the number of bits required for describing the locations of the non-zero entries in each frame does not depend on the selected quantization level (or $\delta$), as $\delta\to 0$, the effect of these  additional bits  becomes negligible. Therefore, the $\alpha$-dimension of the family of codes designed  for the described class of multi-frame signals is equal to $k\over nB$. 

Finally, in Theorem \ref{thm:main-csp}, the measurements are assumed to be noise-free, which is not a realistic assumption. The following theorem shows the robustness of this method to bounded  additive noise. 
\begin{theorem}\label{thm:main-csp-noisy}
Consider the same setup as in Theorem \ref{thm:main-csp}. Assume that the measurements are corrupted by additive noise vector $\zv$. That is, $\yv=\sum_{i=1}^B\Dmat_i\mvec{x}_i+\zv$, where $\zv\in\mathds{R}^n$ denotes the measurement noise and  ${1\over \sqrt{n}}\Lp{\zv}{2}{}\leq \sigma_z$, for some $\sigma_z\geq 0$.  Let $\hat{\xv}$ denote the solution of \eqref{eq:CSP-block}. Assume that $\epsilon>0$ is a free parameter such that    $\epsilon \leq 2 \sqrt{K}$, where $K=8/3$.  Then,
	\begin{align}
	 {1\over \sqrt{nB}}\|\mvec{x}-\hat{\mvec{x}}\|_2 \leq {1\over \sqrt{nB}}\sqrt{\delta} +\rho \epsilon +{2\sigma_z\over \sqrt{B}},
	\end{align}
	with a probability larger than $	 { 1- 2^{nBr+1} {\rm e}^{- {\epsilon^4 n \over 4^3 K^2 }}}$. 
\end{theorem}
The proof is presented in Section \ref{append:main_csp-noisy}. 
%~~~~~~~~~~~~~~~~~~~~~~~%~~~~~~~~~~~~~~~~~~~~~~~%~~~~~~~~~~~~~~~~~~~~~~~%~~~~~~~~~~~~~~~~~~~~~~~
%~~~~~~~~~~~~~~~~~~~~~~~%~~~~~~~~~~~~~~~~~~~~~~~%~~~~~~~~~~~~~~~~~~~~~~~%~~~~~~~~~~~~~~~~~~~~~~~
\section{Efficient  compression-based snapshot CS} \label{Sec:cbSCS}
In the previous section we discussed a compression-based recovery method for snapshot compressed sensing which was inspired by the CSP optimization. Finding the solution of this optimization requires solving a high-dimensional non-convex discrete optimization, $\min_{\mvec{c}\in{\cal C}}\|\mvec{y}-\sum_{i=1}^B\Dmat_i\mvec{c}_i\|_2^2$. Solving this optimization  involves minimizing a convex cost function over exponentially many codewords. Hence,  finding the solution of the CSP optimization  solution through exhaustive search over the codebook is infeasible, even for small values of blocklength $n$. To address this issue, in the following, we propose two different iterative algorithms for  compression-based snapshot CS that are both computationally efficient, and both achieve good performances. 

\subsection{Recovery Algorithm: Compression-based projected gradient descent}\label{Sec:PGD}

Projected gradient descent (PGD) is a well-established method for solving convex optimization problems and there have been extensive studies on the convergence performance of this algorithm \cite{nesterov2013introductory}. More recent results, such as \cite{attouch2013convergence}, also explore the performance of such algorithms when applied to non-convex problems different from those studied in this paper.

Inspired by PGD, Algorithm~\ref{algo:PGD} described below is an iterative algorithm designed to approximate the solution of the non-convex optimization described in  \eqref{eq:CSP-block}. Each iteration involves two key steps: 
\begin{itemize}
	\item [i)] moving in the direction of the gradient of the cost function,
	\item [ii)] projecting the result onto the set of codewords.
\end{itemize}
Note that both steps are computationally very efficient.  The gradient descent step involves matrix-vector multiplication, \ie,   $\Hmat \xv^t$ and $ \Hmat\ts \ev^t$  with $\ev^t=\yv - \Hmat\xv^t$ and $\Hmat=[\Dmat_1,\ldots,\Dmat_B]$. The second step, the projection on the set of codewords, can be performed by applying the encoder and the decoder of the compression code. 

\begin{algorithm}[htbp!]
	\caption{CbPGD for Snapshot CS recovery}
	\begin{algorithmic}[1]
		\REQUIRE$\Hmat$, $\yv$.
		\STATE Initial $\mu>0$, $\xv^0 = 0$.
		\FOR{$t=0$ to Max-Iter }
		\STATE Calculate: $\ev^t = \yv - \Hmat \xv^t$.
		\STATE Projected gradient descent:  $\sv^{t+1} = \xv^{t} + \mu \Hmat\ts \ev^t$.
		\STATE Projection via compression:  $\xv^{t+1} = g(f(\sv^{t+1}))$.
		\ENDFOR
		\STATE $\textbf{Output:}$ Reconstructed signal $\hat \xv$.
	\end{algorithmic}
	\label{algo:PGD}
\end{algorithm}

%\subsection{Noiseless Case}
The following theorem characterizes the performance of the proposed compression-based PGD (CbPGD) algorithm under the noiseless case, \ie, $\zv = 0$ in Eq.~\eqref{Eq:Hx+z}, and shows that if $B$ is small enough, Algorithm \ref{algo:PGD} converges. 
\begin{theorem}\label{thm:main-PGD}
		Consider a compact set ${\cal Q}\subset\mathbb{R}^{nB}$, such that for all $\mvec{x}\in{\cal Q}$, $\Lp{\mvec{x}}{\infty}{}\leq {\rho\over 2}$. Furthermore, consider a compression code for set ${\cal Q}$ with encoding and decoding mappings, $f$ and $g$, respectively.  Assume that the code  operates at rate $r$ and distortion $\delta$, and Assumption \ref{assumption:1} holds. 	
		Consider $\mvec{x}\in\cal{Q}$, and let $\tilde{\mvec{x}}=g(f(\mvec{x}))$. Assume that $\mvec{x}$ is measured as $\mvec{y}=\sum_{i=1}^B\Dmat_i\mvec{x}_i$, where $\Dmat_i=\diag(D_{i1},\ldots,D_{in})$, and $D_{ij}\stackrel{\rm i.i.d.}{\sim} {\cal N}(0,1)$. Let $K=8/3$ and assume that $\delta\leq 2K\rho^2$. Set $\mu=1$, and let ${\mvec{x}}^{t}$ denote the output  of Algorithm \ref{algo:PGD} at iteration $t$.
		Then, given  $\lambda\in(0,0.5)$, for $t=0,1,\ldots$, either ${1\over nB}\Lp{\tilde{\mvec{x}}-{\mvec{x}}^{t}}{2}{2}\leq  \delta$,  or 
		%\vspace{-2mm}
		\begin{align}
		{{1\over \sqrt{nB}} \Lp{{\mvec{x}}^{t+1}-\tilde{\mvec{x}}}{2}{}\leq {2\lambda \over \sqrt{nB}}\Lp{{\mvec{x}}^{t}-\tilde{\mvec{x}}}{2}{} + 4\sqrt{\delta}},\label{eq:iterations-thm2}
		\end{align}
		with a probability at least 
		%\vspace{-2mm}
		\begin{equation}
		{1- 2^{4nBr}{\rm e}^{-({\delta\over 2K\rho^2})^2\lambda^2 n}  - (2^{2nBr}+1) {\rm e}^{-n (\frac{\delta}{2K\rho^2})^2}}. \label{eq:prob-error-thm2}
		\end{equation}	
\end{theorem}
The proof is presented in Section \ref{append:main_pgd}. The following direct  corollary of Theorem \ref{thm:main-PGD} shows how for a bounded number of frames $B$, determined by the properties of the compression code, the algorithm converges with high probability. 
\begin{corollary}\label{coro:main-PGD}
Consider the same setup as Theorem \ref{thm:main-PGD}. Given    $\lambda\in(0,0.5)$ and $\epsilon>0$, assume that 
		\begin{align}
		B\leq {1+\epsilon\over 100 r} ({\delta \lambda\over \rho^2})^2.
		\end{align}
		Then, for $t=0,1,\ldots$, either ${1\over nB}\Lp{\tilde{\mvec{x}}-{\mvec{x}}^{t}}{2}{2}\leq  \delta$,  or 
		%\vspace{-2mm}
		\[
		{{1\over \sqrt{nB}} \Lp{{\mvec{x}}^{t+1}-\tilde{\mvec{x}}}{2}{}\leq {2\lambda \over \sqrt{nB}}\Lp{{\mvec{x}}^{t}-\tilde{\mvec{x}}}{2}{} + 4\sqrt{\delta}},
		\]
		with a probability larger than
		%\vspace{-2mm}
		\begin{equation}
		1- {\rm e}^{-({3\delta\lambda\over 16\rho^2})^2 \epsilon n}.
				\end{equation}	
\end{corollary}
To better understand the convergence behavior of Theorem \ref{thm:main-PGD}, the following corollary directly bounds the error at time $t$ as a function of the initial error and the distortion of the compression code. 

\begin{corollary}\label{coro:main-PGD-error-bound}
Consider the same setup as Theorem \ref{thm:main-PGD}. At iteration $t$, $t=0,1,\ldots,$, define the normalized error as 
\[
e_t\triangleq {1\over \sqrt{nB}} \Lp{{\mvec{x}}^{t}-\tilde{\mvec{x}}}{2}{}.
\]
Consider    $\lambda\in(0,0.5)$, initialization  point $\xv^0$, and $\epsilon>0$.  Assume that $B\leq {1+\epsilon\over 100 r} ({\delta \lambda\over \rho^2})^2$.  Then at iteration $t$, either   $e_{t'}\leq  \sqrt{\delta}$, for some $t'\in\{1,\ldots,t\}$, or 
		%\vspace{-2mm}
		\[
		e_{t+1}\leq (2\lambda)^{t+1}e_0+ {4\over 1-2\lambda}\sqrt{\delta},
		\]
		with a probability larger than $1- {\rm e}^{-({3\delta\lambda\over 16\rho^2})^2 \epsilon n}$.
\end{corollary}
%
%\textcolor{blue}{
%\begin{itemize}
%\item Theorem~\ref{thm:main-PGD} proves that the final distance between the recovered signal and $\tilde\xv$ is bounded by a function of $\lambda$ and $\delta$, where  $\tilde\xv$ is the decoded version of the ground truth signal $\xv$ using the encoding and decoding mappings, $(f,g)$. 	
%%
%\item In Theorem~\ref{thm:main-PGD}, the requirement that ${1\over nB}\|\tilde{\mvec{x}}-{\mvec{x}}^{t}\|_2^2\geq  \delta$, while important for the proof, is a trivial assumption,  as the error introduced by the compression code itself is $\delta$. So using a lossy compression code with distortion $\delta$, it is expected that the final distortion is larger than $\delta$. 
%%
%\item Finally,  in Theorem~\ref{thm:main-PGD}, the step-size is set to  $\mu=1$. However, as discussed in later in Section~\ref{Sec:simulations}, this may not be an ideal  choice in practice.  In fact, in our simulations,   we tune this step-size in each iteration.
%\end{itemize}}	

%\subsubsection{Noisy Case}

Next we consider the case where the measurements are corrupted by additive white Gaussian noise. The following theorem  analyzes the convergence guarantee of the CbPGD algorithm in the  case of noisy  measurements and proves its robustness. 

\begin{theorem}\label{thm:main-PGD_noise}
 Consider the same setup as Theorem \ref{thm:main-PGD}. Further assume that  the measurements are corrupted by additive noise as
\[
\mvec{y}=\sum_{i=1}^B {\Dmat}_i\mvec{x}_i+\zv,
\]
where $\zv\in\mathds{R}^n$ and $\{z_i\}_{i=1}^n\stackrel{i.i.d.}{\sim}\N(0,\sigma^2)$. 
Then, given    $\lambda\in(0,0.5)$ and $\epsilon_z\in(0,\sqrt{\rho})$,  for $t=0,1,\ldots$, either ${1\over nB}\Lp{\tilde{\mvec{x}}-{\mvec{x}}^{t}}{2}{2}\leq  \delta$,  or 
\[
{1\over \sqrt{nB}} \Lp{{\mvec{x}}^{t+1}-\tilde{\mvec{x}}}{2}{}\leq {2\lambda \over \sqrt{nB}}\Lp{{\mvec{x}}^{t}-\tilde{\mvec{x}}}{2}{} + 4\sqrt{\delta}+{2\epsilon_z\sigma\over \sqrt{B}},
\]
with a probability larger than
\begin{align}
1- 2^{4nBr}{\rm e}^{-({{3\delta}\over {16\rho^2}})^2\lambda^2 n}  &- (2^{2nBr}+1) {\rm e}^{-n ({{3\delta}\over {16\rho^2}})^2}%\nonumber\\
%&
- 2^{2nBr}{\rm e}^{- n(\frac{3\epsilon_z }{16\rho} )^2\delta}.
\end{align}
%where $K=8/3$.	
\end{theorem}

	The proof is presented in Section~\ref{append:main_pgd_noise}.

\begin{remark}
The  contribution of the measurement  noise in Theorem \ref{thm:main-PGD_noise} can be seen in two terms. First, there is an  additional error term, ${2\epsilon_z\sigma\over \sqrt{B}}$, which is proportional to the power of the noise. Secondly,  the term $2^{2nBr}{\rm e}^{- n(\frac{3\epsilon_z }{16\rho} )^2\delta}$, which   is part of the error probability, also depends on noise.  In order to reduce the effect of noise, the first term suggests that we need to decrease the sampling rate, or equivalently increase $B$. This is of course counter-intuitive. However, note that this is happening because the non-zero entries of the sensing matrix are drawn i.i.d.~$\N(0,1)$, and therefore the signal-to-noise ratio (SNR) per measurement is also proportional to $B$. To consider the more realistic situation, one can change the power of noise to $B\sigma^2$, i.e. fix the SNR with respect to $B$. Then, we see that the true effect of noise (in addition to increasing the reconstruction error) is revealed in $2^{2nBr}{\rm e}^{- n(\frac{3\epsilon_z }{16\rho} )^2\delta}$. This term puts an extra upper bound on achievable $B$, the number of frames that can be combined together and later recovered.  

%This term shows that Note that each such term in the error probability (there are three of them) puts an upper bound on $B$, number of frames that can be combined together and later recovered. (Recall that, given the structure of the sensing matrix, $1/B$ is equal to the sampling  rate in the studied snapshot CS systems.) Therefore, having noise in the system, potentially increases the required sampling rate.
%	It can be seen from Theorem~\ref{thm:main-PGD} and Theorem~\ref{thm:main-PGD_noise} that in the noisy case, the convergence rate of CbPGD is slower and it is dependent on the power of the noise; meanwhile, the probability is also smaller than that of noiseless case.
\end{remark}

%\textcolor{blue}{In the next section, we provide an alternative recovery algorithm for snapshot CS, for which, even in our simulations we fix the step-size. }
%~~~~~~~~~~~~~~~~~~~~~~~~~~%~~~~~~~~~~~~~~~~~~~~~~~~~~%~~~~~~~~~~~~~~~~~~~~~~~~~~%~~~~~~~~~~~~~~~~~~~~~~~~~~
%~~~~~~~~~~~~~~~~~~~~~~~~~~%~~~~~~~~~~~~~~~~~~~~~~~~~~%~~~~~~~~~~~~~~~~~~~~~~~~~~%~~~~~~~~~~~~~~~~~~~~~~~~~~

\subsection{Compression-based generalized alternating projection} \label{Sec:GAP}
In the previous section we studied convergence performance  of the CbPGD algorithm for a fixed $\mu$. However, in practice,  as discussed later in Section \ref{Sec:simulations}, the  step size $\mu$ needs to be optimized at each iteration. This optimization  is usually   time-consuming and hence noticeably  increases the run time of the algorithm. 
In order to mitigate this issue, and due to the special structure of the sensing matrix, which makes $\Hmat\Hmat\ts$  a diagonal matrix, inspired by the generalized alternating projection (GAP) algorithm ~\cite{Liao14GAP},  we propose a  compression-based GAP (CbGAP) recovery algorithm for snapshot CS in Algorithm~\ref{algo:GAP}.\footnote{In GAP, the use of $\Hmat\ts\Hmat$ shares the spirit of preconditioning in optimization, which is also discussed in \cite{Yuan18SP}.} Here, as before,  $\Hmat=[\Dmat_1,\ldots,\Dmat_B]$,  where $\Dmat_i=\diag(D_{i1},\ldots,D_{in})$ denotes the sensing matrix. Matrix $\Rmat$ is defined as %$\Rmat = \Hmat\Hmat\ts = {\rm diag}(R_1, \dots, R_n)$, 
%\vspace{-2mm}
\begin{equation}
{\Rmat = \Hmat\Hmat\ts = {\rm diag}(R_1, \dots, R_n)},\label{eq:R}
\end{equation}
where $ R_{j} = \sum_{i=1}^{B} D^2_{ij}, \forall j = 1,\dots,n$. 
%%\vspace{-2mm}
Note that the $\Rmat\inv\ev^t$ in the Euclidean projection step of Algorithm~\ref{algo:GAP} can be computed element-wise and thus is very efficient.
Moreover, during the implementation, we never store $\{\Dmat_i\}_{i=1}^B$ and $\Rmat$, but only their diagonal elements.

\begin{algorithm}[H]
	\caption{CbGAP for Snapshot CS recovery}
	\begin{algorithmic}[1]
		\REQUIRE$\Hmat$, $\yv$.
		\STATE Initial $\mu>0$, $\xv^0 = 0$.
		\FOR{$t=0$ to Max-Iter }
		\STATE Calculate: $\ev^t = \yv - \Hmat \xv^t$.
		\STATE Euclidean projection:  $\sv^{t+1} = \xv^{t} + \mu\Hmat\ts \Rmat\inv \ev^t $. % \label{eq:GAP_updates}
		\STATE Projection via compression:  $\xv^{t+1} = g(f(\sv^{t+1}))$.
		\ENDFOR
		\STATE $\textbf{Output:}$ Reconstructed signal $\hat \xv$.
	\end{algorithmic}
	\label{algo:GAP}
	%%\vspace{-3mm}
\end{algorithm}

The following theorem characterizes the convergence performance of the described compression-based GAP algorithm. Similar to Theorem \ref{thm:main-PGD}, the following theorem proves that if $B$ is small enough, Algorithm  \ref{algo:GAP} converges. 
\begin{theorem}\label{thm:main-GAP}
	Consider the same setting as Theorem \ref{thm:main-PGD}.
	For $t=0,1,\ldots$, let $\sv^{t+1} = \xv^{t} + B\Hmat\ts \Rmat\inv (\yv - \Hmat \xv^t),$ and $\xv^{t+1} = g(f(\sv^{t+1}))$, where $\Rmat = \Hmat\Hmat\ts $. 
	Then, given    $\lambda\in(0,0.5)$, for $t=0,1,\ldots$, either  ${1\over nB}\Lp{\tilde{\mvec{x}}-{\mvec{x}}^{t}}{2}{2}\leq  \delta$, or 
	%%\vspace{-1mm}
	\[
	{{1\over \sqrt{nB}} \|{\mvec{x}}^{t+1}-\tilde{\mvec{x}}\|_2\leq {2\lambda \over \sqrt{nB}}\Lp{{\mvec{x}}^{t}-\tilde{\mvec{x}}}{2}{} + 4\sqrt{\delta}},
	\]
	with a probability at least 
	%\vspace{-2mm}
	\begin{equation}
	 {1- 2^{4nBr}{\rm e}^{-{\lambda^2\delta^2 n\over 2B\rho^4}}-2^{2nBr} {\rm e}^{-{n\delta\over 2\rho^2 B^2}}}.\label{eq:prob-error-thm3}
	\end{equation}		
\end{theorem}

The proof is presented in Section \ref{append:main_gap}. Note that, analogous to Corollary \ref{coro:main-PGD-error-bound}  of Theorem \ref{thm:main-PGD}, Theorem \ref{thm:main-GAP} implies that ${1\over \sqrt{nB}} \|{\mvec{x}}^{t+1}-\tilde{\mvec{x}}\|_2$, for $B$ small enough, is bounded  by ${(2\lambda)^{t+1} \over \sqrt{nB}}\Lp{{\mvec{x}}^{0}-\tilde{\mvec{x}}}{2}{} + {4\over 1-2\lambda}\sqrt{\delta}$ with high probability

\begin{itemize}%{\labelitemi}{\leftmargin=12pt \topsep=-2pt \parsep=0pt}	
	\item 
	Theorem~\ref{thm:main-PGD} and Theorem~\ref{thm:main-GAP} show that  CbGAP and CbPGD  have very similar convergence behaviors. Moreover, the important message of both results is the following: for a fixed $\lambda>0$,  \eqref{eq:prob-error-thm2} and \eqref{eq:prob-error-thm3} bound the number of frames ($B$) that can be combined together, and still be recovered by the CbPGD algorithm and the  CbGAP algorithm, respectively, as a function of $\lambda$, $\delta$, $r$, and $\rho$.
	\item
	Though Theorem~\ref{thm:main-GAP} proves the convergence of GAP when $\mu=B$, in our simulations and in real applications, we found that $\mu \in \{1,2\}$ always leads to better results. By contrast, in CbPGD, $\mu = 2/B$ is usually a good choice for a fixed step-size.

\end{itemize}

%In the noisy case, it is straightforward to drive the convergence guarantee of CbGAP similar to Theorem~\ref{thm:main-PGD_noise}, and is thus omitted here. 
%Xin: I'm not sure it is so straightforward. So let's not make such a claim.

%~~~~~~~~~~~~~~~~~~~~~~~%~~~~~~~~~~~~~~~~~~~~~~~%~~~~~~~~~~~~~~~~~~~~~~~%~~~~~~~~~~~~~~~~~~~~~~~
%~~~~~~~~~~~~~~~~~~~~~~~%~~~~~~~~~~~~~~~~~~~~~~~%~~~~~~~~~~~~~~~~~~~~~~~%~~~~~~~~~~~~~~~~~~~~~~~

\section{Simulation results} \label{Sec:simulations}
As mentioned earlier,  snapshot CS is used in various applications.  
As a widely-used example,  in this section, we report our simulation results for video CS and  compare the performances of  our proposed compression-based PGD and GAP algorithms    with  leading algorithms\footnote{Most recently, recovery methods based on deep neural networks (DNNs)  have also been employed for video CS in general and also snapshot CS~\cite{Iliadis18DSPvideoCS,2016arXivVideoCS}. While such methods show promising performance, in this paper, given that our focus has been on theoretical analysis of snapshot CS systems and on developing theoretically-analyzable algorithms, and given the challenges in heuristically setting the meta-parameters of  DNNs, we have  skipped comparing our results with those methods. We believe that DNN-based recovery methods, potentially, can provide  effective solutions for snapshot CS, however, lacking theoretical tractability,  such approaches are not the focus of this work.}, i.e., GMM~\cite{Yang14GMM,Yang14GMMonline}, GAP-wavelet-tree~\cite{Yuan14CVPR}, and GAP-TV~\cite{Yuan16ICIP_GAP}. For each method, we have used the codes provided by the   authors on their websites. All algorithms are implemented  in MATLAB.

Throughout the simulations, the pixel values are normalized into $[0,1]$, which corresponds to $\rho =2$ in our theorems.
The standard peak-signal-to-noise-ratio (PSNR) is employed as the metric to compare different algorithms.

One key advantage of our proposed snapshot CS recovery algorithms is that they can readily  be combined with any (off-the-shelf or newly-designed) compression codes. In the next sections, we first explore the performance of our proposed methods when MPEG coder ~\cite{LeGall91MPEG} is used as the video compression algorithm of choice. As shown in Section \ref{sec:mpeg}, this approach   marginally improves the recovery performance compared to the existing methods. In Section \ref{sec:new-coder}, we propose employing a customized compression algorithm. Combining this compression code with our proposed compression-based recovery algorithms significantly improves the performance achieving  a PSNR gain of 1.7 to 6 (dB), both in  noiseless and noisy settings.

\begin{figure}[H]
\centering
		\includegraphics[width=.6\textwidth]{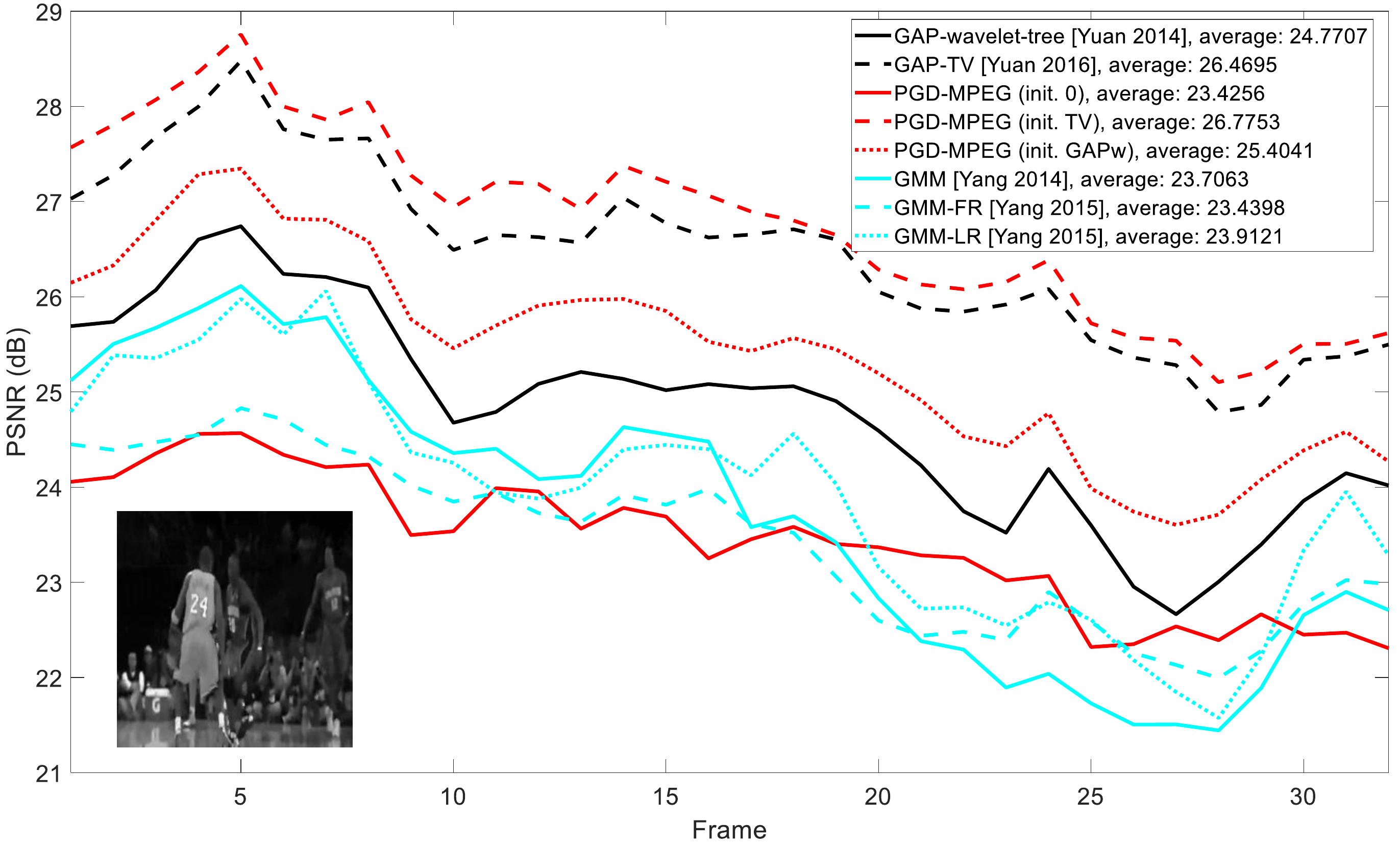}
		%\vspace{-6mm}
		\caption{ PSNR curves of reconstructed video frames compared with ground truth using different algorithms. $B=8$ with 4 measurements. }
		\label{fig:PSNR_kobe}
\end{figure}

\begin{figure}[htbp!]
		\centering
		\includegraphics[width=.6\textwidth]{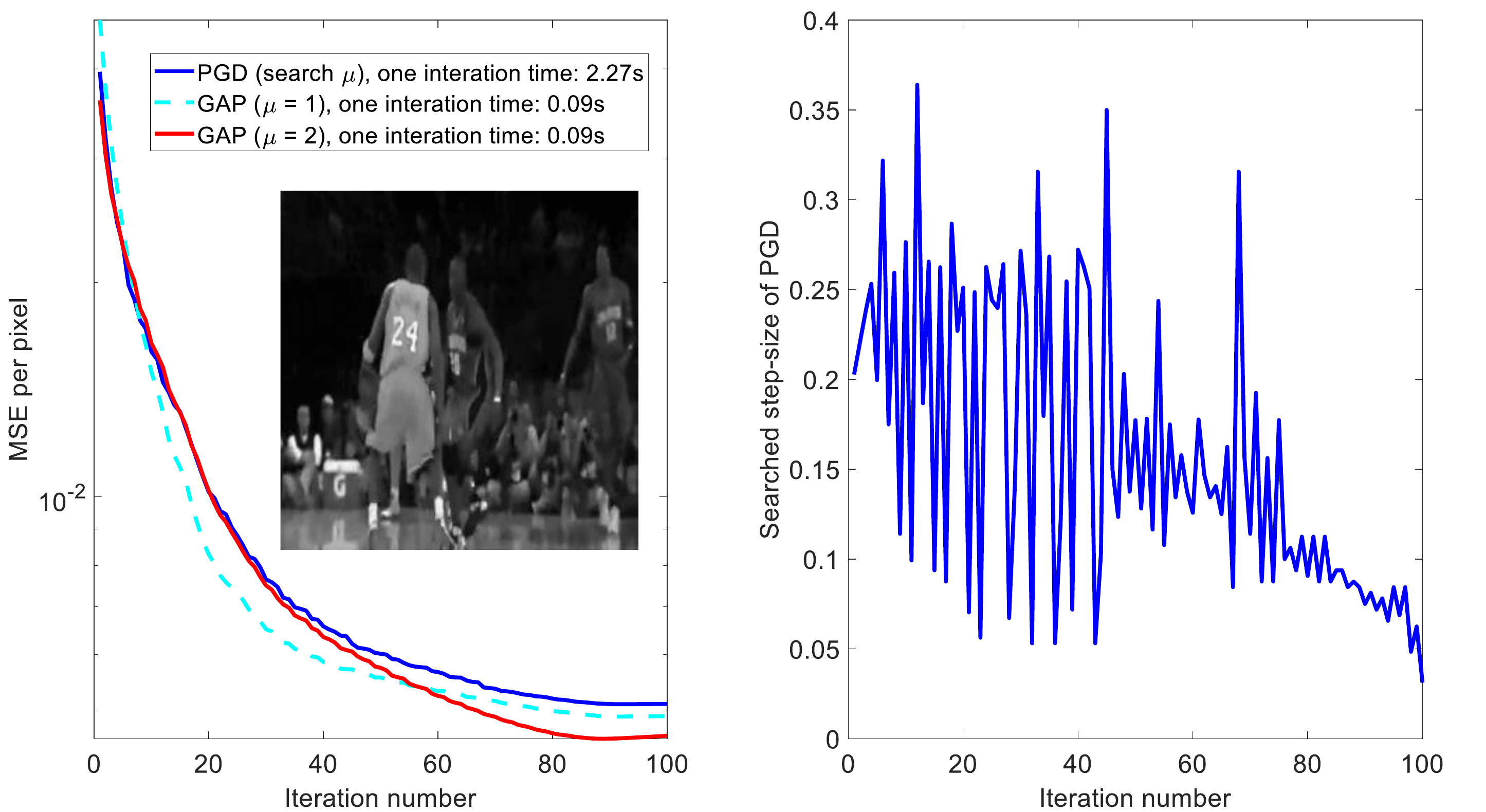}
		\caption{Left:  The MSE of reconstructed video frames using GAP and PGD at each iteration. Right: The searched step size of PGD. }
		\label{fig:kobe_stepsize}
	\end{figure}

For Algorithm \ref{algo:PGD}, Theorems \ref{thm:main-PGD} and \ref{thm:main-PGD_noise} assume that $\mu$ is set to one. However, to speed up the convergence of the algorithm, in our simulations we adaptively choose the step size at each iteration, such that  the  measurement error is minimized after  the projection step.  Specifically, let $\mu_t$ denote the step-size at  iteration $t$. Then, $\mu_t$ is set by solving   the following optimization: 
\begin{equation}
\mu_t = \argmin_{\mu}\left\|\yv - \Hmat g\left(f(\xv^t + \mu \Hmat\ts (\yv - \Hmat \xv^t))\right)\right\|_2.
\end{equation}
%\vspace{-2mm}
This procedure attempts to move the next estimate  as close as possible to the subspace of signals satisfying  the measurement constraints, i.e., ${\cal M} = \{\xiv| \yv = \Hmat \xiv\}$. We employ derivative-free methods, such as~\cite{NelderMead65}, to solve this optimization problem. 
However, this is still time-consuming. Unlike Algorithm \ref{algo:PGD}, Algorithm \ref{algo:GAP} does not entail optimizing the step size and hence runs much faster.

\subsection{MPEG Video Compression}\label{sec:mpeg}
In the first set of experiments we use the MPEG  algorithm as the compression code required by  the CbPGD algorithm (Algorithm \ref{algo:PGD}) and the CbGAP algorithm (Algorithm \ref{algo:GAP}). We  refer to the resulting recovery methods as ``PGD-MPEG" and ``GAP-MPEG", respectively. 
%Different initialization approaches are used in our algorithms, where
%``PGD-MPEG (init. 0)" denotes initializing with zeros, ``PGD-MPEG (init. TV)" signifies the GAP-TV results are used as initialization, and ``PGD-MPEG (init. GAPw)" uses the  GAP-wavelet-tree results for initialization. 
%%
Fig.~\ref{fig:PSNR_kobe} plots the PSNR curves of different algorithms versus the frame number on the {\texttt{Kobe}} dataset used in \cite{Yang14GMM}. 
Each video frame consists of $256\times 256$ pixels. $B = 8$ video frames are modulated and collapsed into a single $256\times 256$  snapshot measurement. For the {\texttt{Kobe}} dataset, there are in total 32 frames and thus 4 measured frames are available. 
For each measured frame, given the masks,  i.e., the sensing matrices $\{\Dmat_i\}_{i=1}^B$, which are generated once and used in all algorithms, the task is  to reconstruct the eight video frames. While
the GMM-based algorithms are typically very slow, GAP-TV ~\cite{Yuan16ICIP_GAP} provides a decent result in a few seconds. Therefore, it is
reasonable to initialize our PGD-MPEG and GAP-MPEG algorithms by the results of GAP-TV. It can be seen in Fig.~\ref{fig:PSNR_kobe} that after such initialization, the compression-based method outperforms other methods, but not with a significant margin. 

Furthermore, note that both CbGAP and CbPGD algorithms are trying to approximate the solution of the non-convex optimization described in  \eqref{eq:CSP-block}. On the other hand, our theorems do not guarantee convergence of either of the algorithms to the global optimal solution. Instead, our theoretical results guarantee that, with a high probability, each method convergences to a point that is in the close vicinity of the desired codeword. This is also confirmed by our simulation results. As seen in Fig.~\ref{fig:PSNR_kobe}, regardless of the starting point, the PGD-MPEG algorithm converges and achieves a decent performance. However, changing the initialization point clearly affects the performance and a good initialization can noticeably improve the final result. 

In Fig.~\ref{fig:kobe_stepsize}, we plot the  average per-pixel reconstruction mean square error (MSE) of (fixed step-size) GAP-MPEG  and  (step-size-optimized)  PGD-MPEG, as the iterative algorithms proceed.  
It can be observed  that after around 100 iterations GAP-MPEG and PGD-MPEG converge to  similar levels of MSE. Moreover, the figure shows that setting $\mu = 2$,  GAP-MPEG  outperforms both PGD-MPEG and GAP-MPEG with $\mu=1$. 
Since no step size search is required by GAP-MPEG, it runs much faster than PGD-MPEG. In fact, one iteration of GAP-MPEG  on average takes about  0.09 seconds, which is  280 times faster than the time required by each iteration  of PGD-MPEG. Through applying  $\{R_j\}_{j=1}^n$, GAP-MPEG is applying a different step size to each measurement dimension, while PGD-MPEG on the other hand is trying to search for a \emph{fixed}  step size that works well for all measurement dimensions. The simulation results suggest that the former method while computationally more efficient achieves a better performance as well.

Finally, given that CbGAP achieves a similar or even better performance than CbPGD while running considerably faster, in the  experiments done in the next section, we only report the performance results of the CbGAP algorithm. 

\begin{figure}[htbp!]
	%\vspace{-3mm}
	\centering
	\includegraphics[width=.9\textwidth]{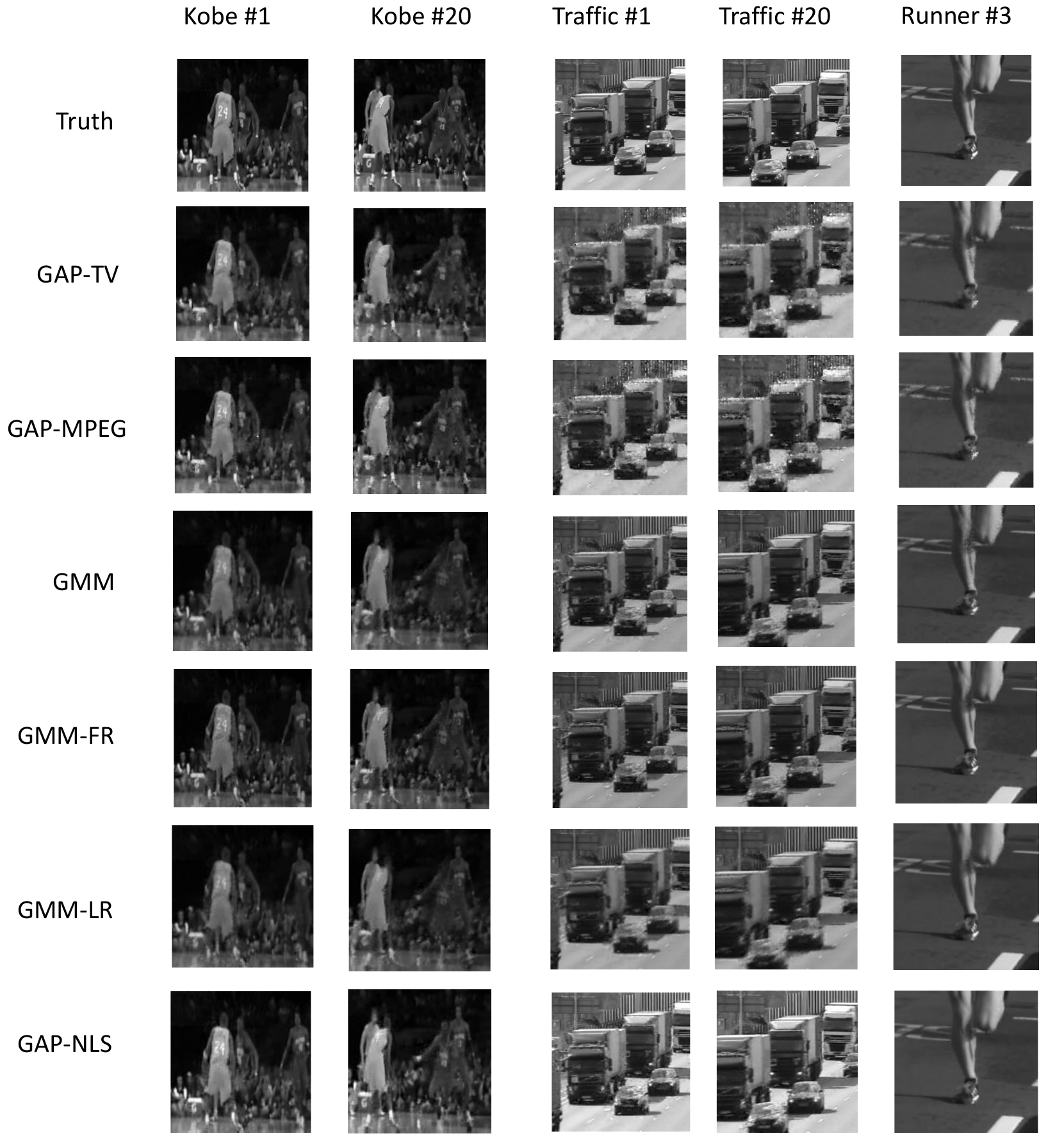}
	%\vspace{-8mm}
	\caption{Reconstructed video frames of three datasets compared with group truth using different algorithms (shown on left). }
	\label{fig:img_3data}
	%\vspace{-2mm}
\end{figure}

\subsection{Customized Video Compression}\label{sec:new-coder}
\begin{table}%{R}{0.5\textwidth}
	%\vspace{-7mm}
	\caption{PSNR ({dB}) of reconstructed videos using different algorithms under noise-free measurements}
	%	\vspace{-2mm}
	\centering
		\begin{tabular}{c|ccc}
			Algorithm   & Kobe & Traffic & Runner \\ 
			\hline GAP-TV~\cite{Yuan16ICIP_GAP} & 26.47 & 20.17 & 30.05 \\
			GAP-wavelet-tree~\cite{Yuan14CVPR} & 24.77 & 21.16 & 25.77 \\
			\hline GMM~\cite{Yang14GMM} & 23.71 & 21.03 & 29.03 \\
			GMM-FR~\cite{Yang14GMMonline}& 23.44 & 21.84 & 33.16 \\
			GMM-LR~\cite{Yang14GMMonline}& 23.91 & 21.09 & 30.70 \\
			\hline 
			GAP-MPEG & 26.78 & 21.78 & 31.02 \\
			GAP-NLS & {\bf 28.70} & {\bf 24.79} & {\bf 37.18} \\
		\end{tabular}
	\label{Table:GAP-NLS}
	%\vspace{-4mm}
\end{table}
MPEG video compression  involves performing  two key steps: JPEG compression of the intra-frames (I-frames) and motion compensation, which
exploits the temporal dependencies between different frames. One drawback of the JPEG-based MPEG compression is that,  since it relies on discrete Cosine transformation (DCT) of individual local patches, it does not exploit  {\em nonlocal} similarities \cite{Buades05NLM} (The most recent video Codec, e.g., H.265/266 considers the similarities between patches but the patches need to be connected, thus it is still local similarity). Using nonlocal similarities in images or videos can potentially  improve the performance significantly and this has been observed in various applications~\cite{Dabov07BM3D,Dong14TIP,Maggioni12vBM4D,Liu18TPAMI}. 

As mentioned earlier, an advantage of compression-based recovery methods is that they can   be combined with any appropriate compression code. Therefore, motivated by the discussed structures based on nonlocal similarities,  we describe a compression code that takes advantage of such structures. Note that to employ  a compression code within either of our algorithms, we  only need to have access to   the combined $g(f(\xv))$ mapping, for any $\xv\in\mathbb{R}^{nB}$, and not $f(\xv)$ itself. Hence, to describe our proposed compression code, we mainly  focus on this mapping from input to its lossy reconstruction. The details of the code can be found in Appendix \ref{App:A}. The key operation of the proposed code is the following:  given a $B$-frame video $\{\xv_i\}_{i=1}^B$ signal, using a small square window, it crawls  over the first frame  and  considers all  3D blocks that result by considering the image of each block across the remaining $(B-1)$ frames. Then the similarities between these 3D blocks are measured and ``similar'' blocks are grouped together. After this, group-sparsity principles~\cite{Dong13TIP,Liu15CVPR} are used to encode such groups of blocks, which shares the same spirit with vBM4D~\cite{Maggioni12vBM4D}.

%Consider an input $B$-frame video $\{\xv_i\}_{i=1}^B$.  The video is  decomposed into small overlapping 3D blocks, and for each reference block, we search the similar blocks across the video and construct a (4D) block groups. After this, we perform 4D DCT to these block groups and zero out the coefficients with small amplitudes.
%The coefficients with large amplitudes ane encoded. The decoder is simply a 4D inverse DCT.
%Via employing this non-local structure, we have dramatically improved the video reconstruction results. 
Combining the described code with the CbGAP method (Algorithm \ref{algo:GAP}) results in an algorithm which we refer to as ``GAP-NLS" (GAP with nonlocal structure). 
%We termed this algorithm as ``GAP-NLS" (GAP with non-local structure), where GAP is used along with this pseudo-encode/decode. We call this ``pseudo" as we did not consider the quantization/dequantization steps.
To evaluate the performance of GAP-NLS, in addition to the {\texttt{Kobe}} dataset used above, we also consider the {\texttt{Traffic}} dataset used in~\cite{Yang14GMM} and the {\texttt{Runner}} dataset from \cite{Runner}.
Table~\ref{Table:GAP-NLS} summarizes the video reconstruction results of our proposed method  (GAP-NLS) compared with GAP-MPEG and other algorithms for these three videos.
It can be observed  that GAP-NLS achieves the best performance; It outperforms  the best PSNR result achieved by other algorithms on  {\texttt{Kobe, Traffic}}, and {\texttt{Runner}} datasets, by  more than \{2.2, 2.9, 4.0\}  (dB),   respectively.
The reconstructed video frames can be found in Fig.~\ref{fig:img_3data}. 
Since the encoder here is more complicated than MPEG, each iteration costs about 3 seconds in our MATLAB implementation. Therefore, if a fast result is desired, GAP-MPEG is recommended and GAP-NLS is the best fit for high accuracy  reconstructions.

\begin{table}[htbp!]%{R}{0.5\textwidth}
	%\vspace{-7mm}
	\caption{PSNR (dB) of reconstructed videos using different algorithms under various noisy measurements}
	%	\vspace{-2mm}
	\small
	\centering
	\resizebox{1\textwidth}{!}
	{
		\begin{tabular}{c|ccc|ccc|ccc}
			& \multicolumn{3}{|c|}{Kobe} & \multicolumn{3}{|c|}{Traffic} & \multicolumn{3}{|c}{Runner} \\ 
			Algorithm   & $\sigma = 0.01$ & $\sigma = 0.1$ & $\sigma = 0.5$ & $\sigma = 0.01$ & $\sigma = 0.1$ & $\sigma = 0.5$ & $\sigma = 0.01$ & $\sigma = 0.1$ & $\sigma = 0.5$ \\ 
			\hline GAP-TV~\cite{Yuan16ICIP_GAP} & 25.97 & 25.33& 18.58 & 20.17& 19.95& 16.40& 29.67& 28.47& 19.17 \\
			GAP-wavelet-tree~\cite{Yuan14CVPR} & 23.77&24.19& 18.12 & 21.15& 20.91& 16.85 &25.61  &23.76&19.16\\
			\hline GMM~\cite{Yang14GMM} & 23.70 & 22.92& 16.83 & 20.98   & 19.51& 11.45& 29.17  &26.83& 16.15\\
			%GMM-FR~\cite{Yang14GMMonline}& 23.44 & 21.84 & 33.16 \\
			%GMM-LR~\cite{Yang14GMMonline}& 23.91 & 21.09 & 30.70 \\
			\hline 
			GAP-MPEG & 26.78 &25.53&19.69& 21.77  & 21.10& 17.05& 30.91 &28.30&21.06\\
			GAP-NLS & {\bf 28.61}&  {\bf 27.04}&  {\bf 20.64} &{\bf 24.56} &{\bf 23.98}&{\bf 18.54}& {\bf 34.57} &{\bf 30.80}& {\bf 22.00}\\
		\end{tabular}
	}
	\label{Table:GAP-Noisy}
	%\vspace{-4mm}
\end{table}

\subsection{Robustness to noise}
So far, in all cases,  the measurements were assumed to be noise-free. However, in real imaging systems, noise is inevitable. In this section, to further investigate the efficacy of our proposed algorithms, we perform the same experiments assuming that the measurements are corrupted by additive white Gaussian noise. 

As mentioned earlier, the pixel values  in the video are normalized into $[0,1]$.  Gaussian noise is added to the measurements, as in ~\eqref{Eq:Hx+z}, where $\zv \sim {\cal N}(0, \sigma^2 \Imat)$. Here $\sigma$ denotes the standard deviation of the Gaussian noise.
Three different values of noise power have been studied, $\sigma = \{0.01, 0.1, 0.5\}$ corresponding to  low, medium and high noise, respectively.  Due to the extremely long running time of GMM-LR and GMM-FR, we hereby only compare with GAP-TV, GAP-wavelet-tree and GMM. The results are summarized in Table~\ref{Table:GAP-Noisy}. 
It can be observed that under small noises ($\sigma = 0.01$), all algorithms work well. However, when the noise is getting larger, the algorithms show different behaviors. The GMM algorithm is affected by the noise severely when $\sigma = 0.5$. This is due to the fact that GMM has a pre-defined noise value in GMM training and testing. However, in real cases, the noise value is unknown and the mismatch between the pre-defined value and real value will defeat the performance.
In every case, our proposed GAP-NLS outperforms the  other algorithms by at least 1.7 (dB) in PSNR. This clearly demonstrate the efficiency of our proposed algorithm.

\begin{remark}
In our theoretical analysis, we  assumed that we are given a compression code with a fixed rate. In practice, most compression codes have a parameter that sets their operating points in terms of rate and distortion. The choice of rate $r$ impacts  the performance of the algorithm. In our simulations, we  heuristically  set the parameters of the compression code to optimize the performance, by sometimes changing it as the algorithm proceeds.  In fact, we believe that one reason the customized compression code achieves better performance is that it provides more flexibility on choosing the quality parameter. Finding a theoretically-motivated approach to setting the parameters is an interesting problem that is left to future work. 
\end{remark}
  
%-----------------------------%-----------------------------%-----------------------------%-----------------------------
\section{Proofs of the main results}\label{Sec:proofs}

%-----------------------------%-----------------------------%-----------------------------%-----------------------------
Before presenting the proof, we first review some known results on the concentration of sub-Gaussian  and sub-exponential random variables. These results are used throughout the  proofs. 

\begin{definition}
	A random variable $X$ is sub-Gaussian when
	\[
	\Lp{X}{\psi_2}{} \triangleq  \inf \LKr{L>0 : \Eox{\exp\LPr{\frac{X^2}{L^2}}} \le 2} < \infty.
	\]
\end{definition}

\begin{definition}
	A random variable $X$ is a sub-exponential random variable, if
	\[
	\Lp{X}{\psi_1}{} \triangleq \inf \LKr{L>0 : \Eox{{\rm e}^{\frac{|X|}{L}}} \le 2} < \infty.
	\]
\end{definition}

\begin{lemma} \label{lemma:basic}
For  a normal random variable $X\sim\N(0,\sigma^2)$,  $\Lp{X}{\psi_2}{}{}=\sqrt{8\over 3}\sigma$.
\end{lemma}
\begin{proof}
Note that $\Eox{{\rm e}^{\frac{X^2}{L^2}}}=\int_{-\infty}^{\infty}(2\pi\sigma^2)^{-0.5}{\rm e}^{-{x^2\over 2}({1\over \sigma^2}-{2\over L^2})}dx=\sqrt{L^2\over L^2-2\sigma^2}$, for $L>2\sigma^2$ and $\Eox{{\rm e}^{\frac{X^2}{L^2}}}=\infty$, otherwise. Moreover,  $\Eox{{\rm e}^{\frac{X^2}{L^2}}}\leq 2$, as long as  $3L^2 \geq  8\sigma^2$. Hence, $\Lp{X}{\psi_2}{}{}=\sqrt{8\over 3}\sigma$. 
\end{proof}

\begin{lemma} \label{lemma:prod_subgaus}
	If  $X$ and $Y$ be sub-Gaussian random variables, then $XY$ is sub-exponential, and $ \Lp{XY}{\psi_1}{} \le \Lp{X}{\psi_2}{}  \Lp{Y}{\psi_2}{}. $
\end{lemma}
\begin{proof}
Let $L= \Lp{XY}{\psi_1}{}$, $L_1=\Lp{X}{\psi_2}{}$ and $ L_2= \Lp{Y}{\psi_2}{}$. To prove the desired result, we need to show that $\Eox{{\rm e}^{\frac{|XY|}{L_1L_2}}}\leq 2$. Since $\frac{|XY|}{L_1L_2}\leq {1\over 2}(\frac{X^2}{L_1^2} + \frac{Y^2}{L_2^2})$, $\Eox{{\rm e}^{\frac{|XY|}{L_1L_2}}}\leq  \Eox{{\rm e}^{{1\over 2}(\frac{X^2}{L_1^2} + \frac{Y^2}{L_2^2})}}=\Eox{{\rm e}^{\frac{X^2}{2L_1^2}}{\rm e}^{ \frac{Y^2}{2L_2^2}}}\leq {1\over 2} \Eox{{\rm e}^{\frac{X^2}{L_1^2}}+{\rm e}^{ \frac{Y^2}{L_2^2}}}\leq 2$. 
\end{proof}

\begin{theorem}[Bernstein Type Inequality, see e.g. \cite{vershynin2010introduction}]
	\label{thm_bound_sub_exp}
	Consider independent random variables   $\{X_i\}_{i=1}^n$, where for $i = 1, \cdots, n$,  $X_i$ is a sub-exponential random variable.  Let $\max_i \Lp{X_i}{\psi_1}{}\le K$, for some $K>0$. Then for every $t\ge 0$ and every $\mvec{w} = \LCr{w_1, \cdots, w_n}^T \in \setR^{n\times 1}$, we have
	\begin{align}
	&\Pr\LPr{\sum_{i=1}^n w_i\LPr{X_i- \Eox{X_i}} \ge t} \nonumber\\
	&\le \exp\left\{-\min \LPr{\frac{t^2}{4K^2\Lp{\mvec{w}}{2}{2}}, \frac{t}{2K\Lp{\mvec{w}}{\infty}{} } } \right\} .
	\end{align}
\end{theorem}

%-----------------------------%-----------------------------%-----------------------------%-----------------------------%-----------------------------
\subsection{Proof of  Theorem~\ref{thm:main-csp} \label{append:main_csp}}
\begin{proof} %[Proof of Theorem \ref{thm:main-csp}]

	Let $\tilde{\mvec{x}}=g(f(\mvec{x}))$. By assumption the code operates at distortion $\delta$.  Hence, ${1\over nB}\Lp{\mvec{x}-\tilde{\mvec{x}}}{2}{2}\leq \delta.$
Moreover, since $\hat{\mvec{x}}=\arg\min_{\mvec{c}\in{\cal C}}\Lp{\mvec{y}-\sum_{i=1}^B\Dmat_i\mvec{c}_i}{2}{2}$ and $\tilde{\mvec{x}}\in{\cal C}$, we have
\[
\Lp{\mvec{y}-\sum_{i=1}^B\Dmat_i\hat{\mvec{x}}_i}{2}{2} \leq \Lp{\mvec{y}-\sum_{i=1}^B\Dmat_i\tilde{\mvec{x}}_i}{2}{2}, 
\]
or, since  $\mvec{y}=\sum_{i=1}^B\Dmat_i\mvec{x}_i$, 
\begin{align}
\Lp{\sum_{i=1}^B\Dmat_i\LPr{\mvec{x}_i-\hat{\mvec{x}}_i}}{2}{2}& \leq \Lp{\sum_{i=1}^B\Dmat_i\LPr{\mvec{x}_i-\tilde{\mvec{x}}_i}}{2}{2}. \label{eq:main-error-ineq}
\end{align}
Given $\epsilon_1>0$ and $\epsilon_2>0$ and $\mvec{x}\in{\mathbb R}^{Bn}$, define events $\E_1$ and $\E_2$ as
\begin{align}
\Big\{&{1\over n}\Big\|\sum_{i=1}^B\Dmat_i\LPr{\mvec{x}_i-{\mvec{c}}_i}\Big\|_2^2\leq {1\over n}\Lp{\mvec{x}-\mvec{c}}{2}{2}+B\rho^2\epsilon_1:%\nonumber\\
\forall \mvec{c}\in\C\Big\} \label{eq:rec_event1}
\end{align}
and 
%\begin{align*}
%$
%\E_2\triangleq
\begin{align}
\Big\{{1\over n}\Big\|\sum_{i=1}^B\Dmat_i\LPr{\mvec{x}_i-{\mvec{c}}_i}\Big\|_2^2\geq {1\over n}\Lp{\mvec{x}-\mvec{c}}{2}{2}-B\rho^2\epsilon_2:%\nonumber\\
\forall \mvec{c}\in\C\Big\},  \label{eq:rec_event2}
\end{align}
%\end{align*}
respectively. Then, conditioned on $\E_1\cap\E_2$, since $\hat{\mvec{x}}\in\C$ and $\tilde{\mvec{x}}\in\C$,  it follows from \eqref{eq:main-error-ineq} that 
\begin{align}
{1\over n}\Lp{\mvec{x}-\hat{\mvec{x}}}{2}{2}\leq {1\over n}\Lp{\mvec{x}-\tilde{\mvec{x}}}{2}{2}+B\rho^2\epsilon_1+ B\rho^2\epsilon_2. \label{eq:main_error_1}
\end{align}
%where $(a)$ follows from Cauchy-Schwartz inequality,
%\[
%D_{\max}\triangleq\max_{(i,j)\in\{1,\ldots,B\}\times\{1,\ldots,n\}} |D_{ij}|,
%\]
%and $(b)$ holds because $d_B(\mvec{x},\tilde{\mvec{x}})\leq \delta$.
In the rest of the proof, we focus on bounding $P(\E_1^c\cup\E_2^c)$. Note that, for a fixed $\mvec{c}\in\C$, 
\begin{align}
\Big\|\sum_{i=1}^B\Dmat_i\LPr{\mvec{x}_i-\tilde{\mvec{c}}_i}\Big\|_2^2
&= \sum_{j=1}^n \LPr{\sum_{i=1}^BD_{ij}(x_{ij}-c_{ij})}^2.
\end{align}
Note that, for  $j=1,\ldots,n$, $\sum_{i=1}^BD_{ij}\LPr{x_{ij}-c_{ij}}$ are independent zero-mean Gaussian random variables. Moreover,
\begin{align}
&\Eox{\LPr{\sum_{i=1}^BD_{ij}(x_{ij}-c_{ij})}^2}\nonumber\\
&=\Eox{\sum_{i=1}^B\sum_{i'=1}^BD_{ij}D_{i'j}(x_{ij}-c_{ij})(x_{i'j}-c_{i'j})}\nonumber\\
&=\sum_{i=1}^B (x_{ij}-c_{ij})^2.
\end{align}
For $j=1,\ldots,n$, define
\begin{align}
w_j\triangleq\sum_{i=1}^B (x_{ij}-c_{ij})^2,\label{eq:def-w}
\end{align}
and
\begin{align}
Z_j\triangleq{\sum_{i=1}^BD_{ij}(x_{ij}-c_{ij})\over \sqrt{\sum_{i=1}^B (x_{ij}-c_{ij})^2}}.\label{eq:def-Z}
\end{align}
Note that $Z_1,\ldots,Z_n$ are i.i.d.~$\N(0,1)$ random variables. Then, $\Lp{\sum_{i=1}^B\Dmat_i\LPr{\mvec{x}_i-\mvec{c}_i}}{2}{2}$ can  be written as $\sum_{j=1}^n w_jZ^2_j$.
%\begin{align}
%\Lp{\sum_{i=1}^B\Dmat_i\LPr{\mvec{x}_i-\mvec{c}_i}}{2}{2}
%&=\sum_{j=1}^n w_jZ^2_j\label{eq:D-error-wZ}
%\end{align}
But
\begin{align*}
\sum_{j=1}^nw_j&=\sum_{j=1}^n\sum_{i=1}^B (x_{ij}-c_{ij})^2
=\Lp{\mvec{x}-\mvec{c}}{2}{2}.
\end{align*}
Therefore, for a fixed $\mvec{c}\in\C$,  
\begin{align}
\Pr&\LPr{{1\over n}\Lp{\sum_{i=1}^B\Dmat_i\LPr{\mvec{x}_i-{\mvec{c}}_i}}{2}{2}\geq {1\over n}\|\mvec{x}-\mvec{c}\|_2^2+B\rho^2\epsilon_1}\nonumber\\
&=\Pr\LPr{{1\over n}\sum_{j=1}^n w_j(Z_j^2-1) \geq B\rho^2\epsilon_1},
\end{align}
where $w_j$ and $Z_j$ are defined in \eqref{eq:def-w} and \eqref{eq:def-Z}, respectively. As proved earlier,  $Z_1,\ldots,Z_n$ are i.i.d.~$\N(0,1)$. Also, for all $j=1,\ldots,n$,  $\Lp{Z^2_j}{\psi_1}{}\le \Lp{Z_j}{\psi_2}{2}\le K$, where  $K={8\over 3}$. Therefore,  it follows from Theorem \ref{thm_bound_sub_exp} that 
\begin{align}
\Pr&\LPr{{1\over n}\Lp{\sum_{i=1}^B\Dmat_i\LPr{\mvec{x}_i-{\mvec{c}}_i}}{2}{2}\geq {1\over n}\Lp{\mvec{x}-\mvec{c}}{2}{2}+B\rho^2\epsilon_1}\nonumber\\
& \leq  \exp\left\{-\min \LPr{\frac{n^2B^2\rho^4\epsilon_1^2}{4K^2\Lp{\mvec{w}}{2}{2}}, \frac{nB\rho^2\epsilon_1}{2K\Lp{\mvec{w}}{\infty}{} } } \right\}. \label{eq:main-conv-epsilon1}
\end{align}
Similarly, for fixed $\mvec{x}$ and $\mvec{c}$,
\begin{align}
\Pr&\LPr{{1\over n}\Lp{\sum_{i=1}^B\Dmat_i\LPr{\mvec{x}_i-{\mvec{c}}_i}}{2}{2}\leq {1\over n}\Lp{\mvec{x}-\mvec{c}}{2}{2}-B\rho^2\epsilon_2}\nonumber\\
& \leq  \exp\left\{-\min \LPr{\frac{n^2B^2\rho^4\epsilon_2^2}{4K^2\Lp{\mvec{w}}{2}{2}}, \frac{nB\rho^2\epsilon_2}{2K\Lp{\mvec{w}}{\infty}{} } } \right\}.\label{eq:main-conv-epsilon2}
\end{align}
The final result follows by the union bound, and the fact that $|\C|\leq 2^{nBr}$.
Note that  $\Lp{\mvec{w}}{1}{}=\sum_{j=1}^nw_j=\sum_{j=1}^n\sum_{i=1}^B (x_{ij}-c_{ij})^2=\Lp{\mvec{x}-\mvec{c}}{2}{2}$. Furthermore,  since by assumption the  $\ell_{\infty}$-norms of all signals in $\cal Q$ are upper-bounded by $\rho/2$, we have
\begin{align}
\Lp{\mvec{w}}{2}{2}&=\sum_{j=1}^nw_j^2
=\sum_{j=1}^n\LPr{\sum_{i=1}^B (x_{ij}-c_{ij})^2}^2\nonumber\\
&\leq \sum_{j=1}^n\LPr{\sum_{i=1}^B \LPr{{\rho\over 2}+{\rho\over 2}}^2}^2
=B^2\rho^4 n,\label{eq:bd-w-ell2-norm}
\end{align}
and
\begin{align}
\Lp{\mvec{w}}{\infty}{}&%=\max_{j=1\ldots,n}|w_j|
=\max_{j=1\ldots,n}\sum_{i=1}^B (x_{ij}-c_{ij})^2
\leq B\rho^2.\label{eq:bd-w-ell-infty-norm}
\end{align}
Therefore, combining \eqref{eq:bd-w-ell2-norm} and \eqref{eq:bd-w-ell-infty-norm} with \eqref{eq:main-conv-epsilon1} and \eqref{eq:main-conv-epsilon2}, it follows that
\begin{align}
\Pr&\LPr{{1\over n}\Lp{\sum_{i=1}^B\Dmat_i\LPr{\mvec{x}_i-{\mvec{c}}_i}}{2}{2}\geq {1\over n}\Lp{\mvec{x}-\mvec{c}}{2}{2}+B\rho^2\epsilon_1} \nonumber\\
&\leq  \exp\LKr{-\min \LPr{\frac{n\epsilon_1^2}{4K^2  }, \frac{n\epsilon_1}{2K } } }, \label{eq:main-conv-epsilon1-simplified}
\end{align}
and
\begin{align}
\Pr&\LPr{{1\over n}\Lp{\sum_{i=1}^B\Dmat_i\LPr{\mvec{x}_i-{\mvec{c}}_i}}{2}{2}\leq {1\over n}\Lp{\mvec{x}-\mvec{c}}{2}{2}-B\rho^2\epsilon_2} \nonumber\\
&\leq  \exp\LKr{-\min \LPr{\frac{n\epsilon_2^2}{4K^2 }, \frac{n\epsilon_2}{2K } } },\label{eq:main-conv-epsilon2-simplified}
\end{align}
respectively. Assume that $\max(\epsilon_1,\epsilon_2)\leq 2K$. Then, $\min \LPr{\frac{\epsilon_1}{2K  },1 }=\frac{\epsilon_1}{2K }$ and $\min \LPr{\frac{B\epsilon_2}{2K  },1 }=\frac{\epsilon_2}{2K }$. Hence, 
\begin{align}
\Pr&\LPr{{1\over n}\Lp{\sum_{i=1}^B\Dmat_i\LPr{\mvec{x}_i-{\mvec{c}}_i}}{2}{2}\geq {1\over n}\Lp{\mvec{x}-\mvec{c}}{2}{2}+B\rho^2\epsilon_1}\nonumber\\
& \leq   \exp\LKr{- \frac{n\epsilon_1^2}{4K^2 }}, \label{eq:main-conv-epsilon1-simplified-final}
\end{align}
and
\begin{align}
\Pr&\LPr{{1\over n}\Lp{\sum_{i=1}^B\Dmat_i\LPr{\mvec{x}_i-{\mvec{c}}_i}}{2}{2}\leq {1\over n}\Lp{\mvec{x}-\mvec{c}}{2}{2}-B\rho^2\epsilon_2} \nonumber\\
&\leq  \exp\LKr{- \frac{n\epsilon_2^2}{4K^2 }}.\label{eq:main-conv-epsilon2-simplified-final}
\end{align}
To finish the proof note that $|\C|\leq 2^{{nBr}}$. Therefore, by the union bound,
\begin{align}
\Pr\LPr{\E_1^c}\leq 2^{nBr} \exp\LKr{ -\frac{n\epsilon_1^2}{4K^2 }},\label{eq:Pe1-thm1}
\end{align}
and 
\begin{align}
\Pr\LPr{\E_2^c}\leq 2^{nBr} \exp\LKr{- \frac{n\epsilon_2^2}{4K^2 }}.\label{eq:Pe2-thm1}
\end{align}
Finally, again by the union bound, $P(\E_1\cap\E_2)\geq 1-P(\E_1^c)-P(\E_2^c)$. Given $0<\epsilon<\frac{16}{3}$,  the desired result follows by letting  $\epsilon_1=\epsilon_2=\epsilon/2$.
Plug this into~\eqref{eq:main_error_1}, we have
\begin{eqnarray}
{1\over n}\Lp{\mvec{x}-\hat{\mvec{x}}}{2}{2}\leq {1\over n}\Lp{\mvec{x}-\tilde{\mvec{x}}}{2}{2}+B\rho^2\epsilon\leq B\delta +B\rho^2\epsilon.
\end{eqnarray}
This is
\begin{equation}
{1\over nB}\Lp{\mvec{x}-\hat{\mvec{x}}}{2}{2}\leq \delta + \rho^2\epsilon.
\end{equation}
Also, from \eqref{eq:Pe1-thm1} and \eqref{eq:Pe1-thm1}, for $\epsilon_1=\epsilon_2=\epsilon/2$, 
\begin{align}
\Pr\LPr{\E_1\cap \E_2}&\geq 1- 2^{nBr} \exp\LKr{ -\frac{n\epsilon_1^2}{4K^2 }} - 2^{nBr} \exp\LKr{ -\frac{n\epsilon_2^2}{4K^2 }}\nonumber\\
%&&\stackrel{\epsilon_1=\epsilon_2=\epsilon/2}{=}
&= 1-2^{nBr+1}\exp\LKr{ -\frac{n\epsilon^2}{16K^2 }}.
\end{align}

\end{proof}

%-----------------------------%-----------------------------%-----------------------------%-----------------------------%-----------------------------
\subsection{Proof of  Theorem~\ref{thm:main-csp-noisy} \label{append:main_csp-noisy}}
\begin{proof} %[Proof of Theorem \ref{thm:main-csp}]

Similar to the proof of Theorem \ref{thm:main-csp}, let $\tilde{\mvec{x}}=g(f(\mvec{x}))$. As before, ${1\over nB}\Lp{\mvec{x}-\tilde{\mvec{x}}}{2}{2}\leq \delta.$ Following the same initial  steps as those used in the proof of Theorem \ref{thm:main-csp}, and noting that by the  the triangle inequality,  $\|\sum_{i=1}^B\Dmat_i\LPr{\mvec{x}_i-\hat{\mvec{x}}_i}+\zv\|\geq \|\sum_{i=1}^B\Dmat_i\LPr{\mvec{x}_i-\hat{\mvec{x}}_i}\|-\|\zv\| $ and $\|\sum_{i=1}^B\Dmat_i\LPr{\mvec{x}_i-\tilde{\mvec{x}}_i}+\zv\|\leq \|\sum_{i=1}^B\Dmat_i\LPr{\mvec{x}_i-\tilde{\mvec{x}}_i}\|+\|\zv\|$, we have
\begin{align}
\Lp{\sum_{i=1}^B\Dmat_i\LPr{\mvec{x}_i-\hat{\mvec{x}}_i}}{2}{}& \leq \Lp{\sum_{i=1}^B\Dmat_i\LPr{\mvec{x}_i-\tilde{\mvec{x}}_i}}{2}{}+2\Lp{\zv}{2}{}. \label{eq:main-error-ineq-noisy}
\end{align}

Given $\epsilon_1>0$ and $\epsilon_2>0$ and $\mvec{x}\in{\mathbb R}^{Bn}$, define events $\E_1$ and $\E_2$ as \eqref{eq:rec_event1} and \eqref{eq:rec_event2}, respectively.
Conditioned on $\E_1\cap\E_2$, since both $\hat{\mvec{x}}$ and $\tilde{\mvec{x}}$ are in the codebook of the code,  it follows from \eqref{eq:main-error-ineq} that 
\begin{align}
\sqrt{\Lp{\mvec{x}-\hat{\mvec{x}}}{2}{2}-nB\rho^2\epsilon_2}\leq \sqrt{ \Lp{\mvec{x}-\tilde{\mvec{x}}}{2}{2}+n B\rho^2\epsilon_1}+2\|\zv\|_2. \label{eq:main_error_1-noise}
\end{align}
But, $\sqrt{a+b}\leq \sqrt{a}+\sqrt{b}$ and $\sqrt{a-b}\geq \sqrt{a}-\sqrt{b}$. Hence, from \eqref{eq:main_error_1-noise}, 
\[
{1\over \sqrt{nB}}\Lp{\mvec{x}-\hat{\mvec{x}}}{2}{}-\sqrt{\rho^2\epsilon_2} \leq {1\over \sqrt{nB}} \Lp{\mvec{x}-\tilde{\mvec{x}}}{2}{}+\sqrt{ \rho^2\epsilon_1}+{2\|\zv\|_2\over \sqrt{n}},
\]
or
  \begin{align*}
{1\over \sqrt{nB}}\Lp{\mvec{x}-\hat{\mvec{x}}}{2}{}&\leq {1\over \sqrt{nB}} \Lp{\mvec{x}-\tilde{\mvec{x}}}{2}{}+\rho( \sqrt{\epsilon_1}+\sqrt{\epsilon_2}) +{2\|\zv\|_2\over \sqrt{nB}}\\
&\leq \sqrt{\delta} +\rho ( \sqrt{\epsilon_1}+\sqrt{\epsilon_2}) +{2\|\zv\|_2\over \sqrt{nB}},
\end{align*}
where the last line follows since the compression code operates at distortion $\delta$. Setting $\epsilon_1=\epsilon_2=\epsilon^2/4$, and noting that by assumption ${\|\zv\|_2\over \sqrt{n}}\leq \sigma_z$, the desired result follows. 
%where $(a)$ follows from Cauchy-Schwartz inequality,
%\[
%D_{\max}\triangleq\max_{(i,j)\in\{1,\ldots,B\}\times\{1,\ldots,n\}} |D_{ij}|,
%\]
%and $(b)$ holds because $d_B(\mvec{x},\tilde{\mvec{x}})\leq \delta$.
Bounding $P(\E_1^c\cup\E_2^c)$ can be done exactly as in  the proof of Theorem \ref{thm:main-csp}. 

\end{proof}

%-----------------------------%-----------------------------%-----------------------------%-----------------------------
\subsection{Proof of  Theorem~\ref{thm:main-PGD} \label{append:main_pgd}}
%\begin{proof}
\begin{proof}
Assume that ${1\over nB}\Lp{\tilde{\mvec{x}}-{\mvec{x}}^{t}}{2}{2}\leq  \delta$ does not hold at iteration $t$. Then to prove the theorem, we need to show that eqref{eq:iterations-thm2} holds. 
	At step $t$, given $\xv^t$, define  the error vector and its normalized version  as
	\begin{align}
	\thetav^t=\tilde{\mvec{x}}-{\mvec{x}}^{t},\label{eq:def-theta}
	\end{align}
	and 
	\begin{align}
	\bar{\mvec{\theta}}^t={\mvec{\theta}^t\over \Lp{\mvec{\theta}^t}{2}{}},\label{eq:def-theta-bar}
	\end{align}
	respectively.
	By this definition, for $i=1,\ldots,B$, the $i$-th block of each of these error vectors can be written as $\ev^t_i= \tilde{\mvec{x}}_i-{\mvec{x}}^{t}_i$,
	and 
	\[
	\bar{\mvec{\theta}}^t_i= {\tilde{\mvec{x}}_i-{\mvec{x}}^{t}_i\over \Lp{\tilde{\mvec{x}}-\mvec{x}^t}{2}{}}.
	\]
	Moreover, since by assumption ${1\over nB}\Lp{\tilde{\mvec{x}}-{\mvec{x}}^{t}}{2}{2}\geq  \delta$, 
	\[
	{1\over nB}\sum_{i=1}^B\Lp{\tilde{\mvec{x}}_i-{\mvec{x}}_i^{t}}{2}{2}\geq  \delta,
	\]
	Therefore,  $\min\LPr{\Lp{\mvec{\theta}^t}{2}{2},\Lp{\mvec{\theta}^{t+1}}{2}{2}}\geq nB\delta$. 
	
	For $i=1,\ldots,B$ and $j=1,\ldots,n$, since for all $\mvec{x}\in{\cal Q}$, $\Lp{\mvec{x}}{\infty}{}\leq {\rho\over 2}$, we have
	\begin{align}
	|\bar{\theta}^t_{ij}|^2&={(\tilde{{x}}_{ij}-{x}^{t}_{ij})^2\over \Lp{\mvec{\theta}^t}{2}{2}}\leq {\rho^2 \over nB\delta}\label{eq:bound-e-bar-ij-k}
	\end{align}
	and, similarly,
	\begin{align}
	|\bar{\theta}^{t+1}_{ij}|^2\leq {\rho^2 \over nB\delta}\label{eq:bound-e-bar-ij-k+1}
	\end{align}
	
	Since  ${\mvec{x}}^{t+1}$ is the closest codeword to $\mvec{s}^{t+1}$ in $\C$, and $\tilde{\mvec{x}}$ is also in $\C$, it follows that
	$\Lp{\mvec{s}^{t+1}-{\mvec{x}}^{t+1}}{2}{2}\leq \Lp{\mvec{s}^{t+1}-\tilde{\mvec{x}}}{2}{2},$
	or 
	\begin{align}
	\sum_{i=1}^B\Lp{\mvec{s}_i^{t+1}-{\mvec{x}}_i^{t+1}}{2}{2}\leq \sum_{i=1}^B \Lp{\mvec{s}_i^{t+1}-\tilde{\mvec{x}}_i}{2}{2}.\label{eq:si-k+1-less-k-first-eq}
	\end{align}
	But 
	\begin{align}
	&\Lp{\mvec{s}_i^{t+1}-{\mvec{x}}_i^{t+1}}{2}{2}=\Lp{\mvec{s}_i^{t+1}-\tilde{\mvec{x}}_i+\tilde{\mvec{x}}_i-{\mvec{x}}_i^{t+1}}{2}{2} \nonumber\\
	&=\Lp{\mvec{s}_i^{t+1}-\tilde{\mvec{x}}_i}{2}{2}+2\LPd{\mvec{s}_i^{t+1}-\tilde{\mvec{x}}_i,\tilde{\mvec{x}}_i-{\mvec{x}}_i^{t+1}}+\Lp{\tilde{\mvec{x}}_i-{\mvec{x}}_i^{t+1}}{2}{2}.
	\end{align}
	%\[\Lp{\mvec{s}_i^{k+1}-{\mvec{x}}_i^{k+1}}{2}{2}=\Lp{\mvec{s}_i^{k+1}-\tilde{\mvec{x}}_i+\tilde{\mvec{x}}_i-{\mvec{x}}_i^{k+1}}{2}{2}=\Lp{\mvec{s}_i^{k+1}-\tilde{\mvec{x}}_i}{2}{2}+2\LPd{\mvec{s}_i^{k+1}-\tilde{\mvec{x}}_i,\tilde{\mvec{x}}_i-{\mvec{x}}_i^{k+1}}+\Lp{\tilde{\mvec{x}}_i-{\mvec{x}}_i^{k+1}}{2}{2}.
	%\]
	Therefore, along with \eqref{eq:si-k+1-less-k-first-eq}, 
	\begin{align}
	&\Lp{\mvec{s}_i^{t+1}-\tilde{\mvec{x}}_i}{2}{2}+2\LPd{\mvec{s}_i^{t+1}-\tilde{\mvec{x}}_i,\tilde{\mvec{x}}_i-{\mvec{x}}_i^{t+1}}+\Lp{\tilde{\mvec{x}}_i-{\mvec{x}}_i^{t+1}}{2}{2} \nonumber\\
	&\le \Lp{\mvec{s}_i^{t+1}-\tilde{\mvec{x}}_i}{2}{2}.
	\end{align}
	This is
	\begin{align}
	&\sum_{i=1}^B\Lp{\tilde{\mvec{x}}_i-{\mvec{x}}_i^{t+1}}{2}{2}\leq 2\sum_{i=1}^B \LPd{\tilde{\mvec{x}}_i-\mvec{s}_i^{t+1},\tilde{\mvec{x}}_i-{\mvec{x}}_i^{t+1}}\nonumber\\
	&\stackrel{(a)}{=} 2\sum_{i=1}^B \LPd{\tilde{\mvec{x}}_i-{\mvec{x}}^{t}_i- \Dmat_i\mvec{e}^t,\tilde{\mvec{x}}_i-{\mvec{x}}_i^{t+1}}\nonumber\\
	&\stackrel{(b)}{=} 2\sum_{i=1}^B \LPd{\tilde{\mvec{x}}_i-{\mvec{x}}^{t}_i,\tilde{\mvec{x}}_i-{\mvec{x}}_i^{t+1}}\nonumber\\
	&\qquad -2 \sum_{i=1}^B \LPd{\Dmat_i\LPr{\mvec{y}-\sum_{j=1}^B\Dmat_j\mvec{x}_j^t},\tilde{\mvec{x}}_i-{\mvec{x}}_i^{t+1}}\nonumber\\
	%	&=2\sum_{i=1}^B \LPd{\tilde{\mvec{x}}_i-{\mvec{x}}^{t}_i,\tilde{\mvec{x}}_i-{\mvec{x}}_i^{t+1}}
	%	-2 \sum_{i=1}^B \LPd{\Dmat_i\LPr{\sum_{j=1}^B\Dmat_j \xv_j-\sum_{j=1}^B\Dmat_j\mvec{x}_j^t},\tilde{\mvec{x}}_i-{\mvec{x}}_i^{t+1}}\nonumber\\
	%	&=2\sum_{i=1}^B \LPd{\tilde{\mvec{x}}_i-{\mvec{x}}^{t}_i,\tilde{\mvec{x}}_i-{\mvec{x}}_i^{t+1}}
	%	-2 \sum_{i=1}^B \LPd{\Dmat_i\LPr{\sum_{j=1}^B\Dmat_j(\xv_j-\mvec{x}_j^t)},\tilde{\mvec{x}}_i-{\mvec{x}}_i^{t+1}}\nonumber\\
	%	&= 2\sum_{i=1}^B \LPd{\tilde{\mvec{x}}_i-{\mvec{x}}^{t}_i,\tilde{\mvec{x}}_i-{\mvec{x}}_i^{t+1}}
	%	-2 \sum_{i=1}^B \LPd{\sum_{j=1}^B\Dmat_j\LPr{\mvec{x}_j-\mvec{x}_j^t},\Dmat_i\LPr{\tilde{\mvec{x}}_i-{\mvec{x}}_i^{t+1}}}\nonumber\\
	&= 2\sum_{i=1}^B \LPd{\tilde{\mvec{x}}_i-{\mvec{x}}^{t}_i,\tilde{\mvec{x}}_i-{\mvec{x}}_i^{t+1}}\nonumber\\
	&\qquad-2  \LPd{\sum_{j=1}^B\Dmat_j\LPr{\mvec{x}_j-\mvec{x}_j^t},\sum_{i=1}^B\Dmat_i\LPr{\tilde{\mvec{x}}_i-{\mvec{x}}_i^{t+1}}}\nonumber\\
	&= 2\sum_{i=1}^B \LPd{\tilde{\mvec{x}}_i-{\mvec{x}}^{t}_i,\tilde{\mvec{x}}_i-{\mvec{x}}_i^{t+1}}\nonumber\\
	&\qquad
	-2 \LPd{\sum_{i=1}^B\Dmat_i\LPr{\mvec{x}_i-\tilde{\mvec{x}}_i+\tilde{\mvec{x}}_i-\mvec{x}_i^t}, \sum_{i=1}^B\Dmat_i\LPr{\tilde{\mvec{x}}_i-{\mvec{x}}_i^{t+1}}}\nonumber\\
%	\end{align}
%	\begin{align}
	&= 2\sum_{i=1}^B \LPd{\tilde{\mvec{x}}_i-{\mvec{x}}^{t}_i,\tilde{\mvec{x}}_i-{\mvec{x}}_i^{t+1}}\nonumber\\
	&\qquad
	-2 \LPd{\sum_{i=1}^B\Dmat_i\LPr{\tilde{\mvec{x}}_i-\mvec{x}_i^t}, \sum_{i=1}^B\Dmat_i\LPr{\tilde{\mvec{x}}_i-{\mvec{x}}_i^{t+1}}}\nonumber\\
	&\qquad-2 \LPd{\sum_{i=1}^B\Dmat_i\LPr{\mvec{x}_i-\tilde{\mvec{x}}_i}, \sum_{i=1}^B\Dmat_i\LPr{\tilde{\mvec{x}}_i-{\mvec{x}}_i^{t+1}}}.\label{eq:main-derivation}
	\end{align}
	where $(a)$ and $(b)$ follow from   $\sv^{t+1}_i = \xv^{t}_i +  \Dmat_i \ev^t$ and $\ev^t =\yv - \sum_{i=1}^B \Dmat_i \xv_i^t$, respectively.  Then, using  $\bar{\mvec{\theta}}^t$ and $\bar{\mvec{\theta}}^{t+1}$ defined in \eqref{eq:def-theta-bar}, the first two terms at the end of \eqref{eq:main-derivation} can be written as 
	\begin{align}
	&2\sum_{i=1}^B \LPd{\tilde{\mvec{x}}_i-{\mvec{x}}^{t}_i,\tilde{\mvec{x}}_i-{\mvec{x}}_i^{t+1}}\nonumber\\
	&\quad
	- 2\LPd{\sum_{i=1}^B\Dmat_i\LPr{\tilde{\mvec{x}}_i-\mvec{x}_i^t}, \sum_{i=1}^B\Dmat_i\LPr{\tilde{\mvec{x}}_i-{\mvec{x}}_i^{t+1}}}\nonumber\\
	=&2 \Lp{\mvec{\theta}^{t}}{2}{} \Lp{\mvec{\theta}^{t+1}}{2}{} (\sum_{i=1}^B \LPd{\bar{\mvec{\theta}}^t_i,\bar{\mvec{\theta}}_i^{t+1}}
		- \langle\sum_{i=1}^B\Dmat_i\bar{\mvec{\theta}}^t_i, \sum_{i=1}^B\Dmat_i\bar{\mvec{\theta}}^{t+1}_i\rangle). \label{eq:first2term_err}
	\end{align}
	And
	\begin{align}
	&\sum_{i=1}^B \LPd{\bar{\mvec{\theta}}^t_i,\bar{\mvec{\theta}}_i^{t+1}}
	- \LPd{\sum_{i=1}^B\Dmat_i\bar{\mvec{\theta}}^t_i, \sum_{i=1}^B\Dmat_i\bar{\mvec{\theta}}^{t+1}_i}\nonumber\\
	&=\sum_{j=1}^n\LPr{\sum_{i=1}^B \bar{\theta}^t_{ij}\bar{\theta}_{ij}^{t+1}
		-\LPr{\sum_{i_1=1}^B  D_{i_1j}\bar{\theta}^t_{i_1j}} \LPr{ \sum_{i_2=1}^B D_{i_2j}\bar{\theta}^{t+1}_{i_2j}}}\nonumber\\
	&={1\over n}\sum_{j=1}^n \LPr{n\sum_{i=1}^B \bar{\theta}^t_{ij}\bar{\theta}_{ij}^{t+1}-U_jV_j},\label{eq:der-Ui-Vj}
	\end{align}
	where, for $j=1,\ldots,n$,  random variables $U_j$ and $V_j$ are defined as
	\[
	U_j\triangleq \sqrt{n} \sum_{i=1}^B  D_{ij}\bar{\theta}^t_{ij},
	\] \label{eq:define_U}
	and 
	\[
	V_j\triangleq \sqrt{n}\sum_{i=1}^B D_{ij}\bar{\theta}^{t+1}_{ij},
	\] \label{eq:define_V}
	respectively. 
	
	Thus, Eq.~\eqref{eq:first2term_err}, \ie, the first two terms of \eqref{eq:main-derivation},  becomes, 
	\begin{align}
	&2\sum_{i=1}^B \LPd{\tilde{\mvec{x}}_i-{\mvec{x}}^{t}_i,\tilde{\mvec{x}}_i-{\mvec{x}}_i^{t+1}}\nonumber\\
	&\quad
	- 2\LPd{\sum_{i=1}^B\Dmat_i\LPr{\tilde{\mvec{x}}_i-\mvec{x}_i^t}, \sum_{i=1}^B\Dmat_i\LPr{\tilde{\mvec{x}}_i-{\mvec{x}}_i^{t+1}}}\nonumber\\
	& = 2 \Lp{\mvec{\theta}^{t}}{2}{} \Lp{\mvec{\theta}^{t+1}}{2}{} {1\over n}\sum_{j=1}^n \LPr{n\sum_{i=1}^B \bar{\theta}^t_{ij}\bar{\theta}_{ij}^{t+1}-U_jV_j}.
	\end{align}
	Impose the Cauchy-Schwarz inequality on the third term in \eqref{eq:main-derivation}, \ie,
	\begin{align}
	&\LPd{\sum_{i=1}^B\Dmat_i\LPr{\mvec{x}_i-\tilde{\mvec{x}}_i}, \sum_{i=1}^B\Dmat_i\LPr{\tilde{\mvec{x}}_i-{\mvec{x}}_i^{t+1}}}\nonumber\\
	&\le \left\|\sum_{i=1}^B\Dmat_i\LPr{\mvec{x}_i-\tilde{\mvec{x}}_i}\right\|_2 \left\|\sum_{i=1}^B\Dmat_i\LPr{\tilde{\mvec{x}}_i-{\mvec{x}}_i^{t+1}}\right\|_2 \nonumber\\
	&= \left\|\sum_{i=1}^B\Dmat_i\LPr{\mvec{x}_i-\tilde{\mvec{x}}_i}\right\|_2 \left\|\sum_{i=1}^B\Dmat_i{\thetav_i^{t+1}}\right\|_2.
	\end{align}
	
	Now, ~\eqref{eq:main-derivation} becomes
	\begin{align}
	\sum_{i=1}^B\Lp{\tilde{\mvec{x}}_i-{\mvec{x}}_i^{t+1}}{2}{2}&\le 2 \Lp{\mvec{\theta}^{t}}{2}{} \Lp{\mvec{\theta}^{t+1}}{2}{} {1\over n}\sum_{j=1}^n \LPr{n\sum_{i=1}^B \bar{\theta}^t_{ij}\bar{\theta}_{ij}^{t+1}-U_jV_j} \nonumber\\
	&\;\;+ 2\left\|\sum_{i=1}^B\Dmat_i\LPr{\mvec{x}_i-\tilde{\mvec{x}}_i}\right\|_2 \left\|\sum_{i=1}^B\Dmat_i{\thetav_i^{t+1}}\right\|_2,
	\end{align}
	which is
	\begin{align}
	\left\|\thetav^{t+1}\right\|_2^2 &\le 2 \Lp{\mvec{\theta}^{t}}{2}{} \Lp{\mvec{\theta}^{t+1}}{2}{} {1\over n}\sum_{j=1}^n \LPr{n\sum_{i=1}^B \bar{\theta}^t_{ij}\bar{\theta}_{ij}^{t+1}-U_jV_j} \nonumber\\
	&\;\;+ 2\left\|\sum_{i=1}^B\Dmat_i\LPr{\mvec{x}_i-\tilde{\mvec{x}}_i}\right\|_2 \left\|\sum_{i=1}^B\Dmat_i{\thetav_i^{t+1}}\right\|_2.
	\end{align}
	Dividing both sides by $\|\thetav^{t+1}\|_2$,
	\begin{align}
	\left\|\thetav^{t+1}\right\|_2 &\le 2 \Lp{\mvec{\theta}^{t}}{2}{} {1\over n}\sum_{j=1}^n \LPr{n\sum_{i=1}^B \bar{\theta}^t_{ij}\bar{\theta}_{ij}^{t+1}-U_jV_j} \nonumber\\
	&\;\;+ 2\left\|\sum_{i=1}^B\Dmat_i\LPr{\mvec{x}_i-\tilde{\mvec{x}}_i}\right\|_2 \left\|\sum_{i=1}^B\Dmat_i{\bar \thetav_i^{t+1}}\right\|_2. \label{Eq:simple_2terms}
	\end{align}
	
	In the following, we first bound the first term in \eqref{Eq:simple_2terms} using the Theorem~\ref{thm_bound_sub_exp}.
	
	\begin{itemize}
		\item Bounds using Bernstein type inequality.
		
		Note that $(U_j,V_j)$, $j=1,\ldots,n$,  are i.i.d.~jointly Gaussian, such that 
		\[
		U_j\sim\N\LPr{0,  n \sum_{i=1}^B (\bar{\theta}^t_{ij})^2},
		\]
		\[
		V_j\sim \N\LPr{0,  n \sum_{i=1}^B  (\bar{\theta}^{t+1}_{ij})^2}
		\]
		and
		\[
		\Eox{U_jV_j}=n\sum_{i=1}^B \bar{\theta}^t_{ij}\bar{\theta}_{ij}^{t+1}.
		\]
		By Lemma \ref{lemma:prod_subgaus}, since $U_j$ and $V_j$ are Gaussian random variables, $U_jV_j$ is a sub-exponential random variable and
		\begin{align}
		\Lp{U_jV_j}{\psi_1}{}\leq {8\over 3}n\sqrt{ \sum_{i=1}^B (\bar{\theta}^t_{ij})^2 \sum_{i=1}^B (\bar{\theta}^{t+1}_{ij})^2}.
		\end{align}
		On the other hand, $(\bar{\theta}^t_{ij})^2$ and $(\bar{\theta}^{t+1}_{ij})^2$ are bounded as \eqref{eq:bound-e-bar-ij-k}, \ie,
		$|\bar{\theta}^{t}_{ij}|^2\leq {\rho^2 \over nB\delta}, |\bar{\theta}^{t+1}_{ij}|^2\leq {\rho^2 \over nB\delta}$ Therefore,
		\begin{align}
		\Lp{U_jV_j}{\psi_1}{}&\leq {8\over 3}n \sum_{i=1}^B  {\rho^2\over nB\delta}\nonumber\\
		&= {8\rho^2 \over 3\delta}\nonumber\\
		&= {K\rho^2 \over \delta}, \label{eq:def-K}
		\end{align}
		where $K=8/3$. 
		Define the set of  possible normalized error vectors of interest as follows
		\begin{align}
		\F\triangleq \LKr{{\mvec{c}-\mvec{c}'\over\Lp{\mvec{c}-\mvec{c}'}{2}{} }: (\mvec{c},\mvec{c}')\in\C^2, \Lp{\mvec{c}-\mvec{c}'}{2}{}\geq \sqrt{nB\delta}}. \label{eq:def-set-F}
		\end{align}
		Employing  Theorem \ref{thm_bound_sub_exp} and \eqref{eq:der-Ui-Vj}, for a fixed $(\bar{\mvec{\theta}},\bar{\mvec{\theta}}')\in\F^2$, we have
		\begin{align}
		\Pr&\LPr{\sum_{j=1}^n\LPr{\sum_{i=1}^B \bar{\theta}_{ij}\bar{\theta}'_{ij}
				-\LPr{\sum_{i=1}^B  D_{ij}\bar{\theta}_{ij}} \LPr{ \sum_{i=1}^B D_{ij}\bar{\theta}'_{ij}}}\geq \epsilon }\nonumber\\
		&=\Pr\LPr{{1\over n}\sum_{j=1}^n \LPr{n\sum_{i=1}^B \bar{\theta}_{ij}\bar{\theta}_{ij}-U_jV_j}\geq \epsilon}\nonumber\\
		&\leq \exp\LKr{-\min \LPr{\frac{n\epsilon^2}{4K^2}, \frac{n\epsilon}{2K } } }\nonumber\\
		&= \exp\LKr{-\LPr{\frac{n\epsilon\delta}{2K\rho^2 }}\min \LPr{\frac{\epsilon\delta}{2K\rho^2}, 1 } }. 
		%&= \exp\LKr{-{\frac{n\epsilon^2}{4K^2 }} },
		\label{eq:dev-fixed-theta}
		\end{align}
		%\begin{align}
where since  $w_j= \frac{1}{n}$, $j=1,\ldots,n$, in Theorem \ref{thm_bound_sub_exp},  $\|\wv\|_2^2 = \|\wv\|_\infty = 1/n$. This leads to the final results in~\eqref{eq:dev-fixed-theta}. 
		
		\item Derive the union bound.
		
		% and $(a)$ follows because $\epsilon<2K$ and therefore $\min \LPr{\frac{\epsilon}{2K}, 1 }=\frac{\epsilon}{2K}$. 
		Given $\lambda>0$, define event $\E_1$ as follows
		\begin{align}
		&\E_1\triangleq \LKr{\sum_{j=1}^n\LPr{\sum_{i=1}^B \bar{\theta}_{ij}\bar{\theta}'_{ij}
				-\LPr{\sum_{i=1}^B  D_{ij}\bar{\theta}_{ij}} \LPr{ \sum_{i=1}^B D_{ij}\bar{\theta}'_{ij}}} \nonumber\\
			&\qquad \leq \lambda :\;\forall \;(\bar{\mvec{\theta}}, \bar{\mvec{\theta}}')\in\F^2 }. \label{Eq:eveen_1}
		\end{align}
		Note the $\bar{\theta}_{ij}\bar{\theta}'_{ij}$ in \eqref{Eq:eveen_1} is different  and $\bar{\theta}_{ij}\bar{\theta}'_{ij}$ in \eqref{eq:dev-fixed-theta}, where the one in \eqref{Eq:eveen_1} can be any value in a set $\F$, but the one in \eqref{eq:dev-fixed-theta} is a fixed value.
		
		Let $\{\E_i^c\}_{i=1}^3$ denote the complementary event of $\E_i$, and we need to consider all values in the set $\F$. 
		Combining  the union bound with \eqref{eq:dev-fixed-theta} yields
		\begin{align}
		\Pr\LPr{\E_1^c}&\leq |\F|^2\exp\LKr{-\LPr{\frac{n\lambda\delta}{2K\rho^2 }}\min \LPr{\frac{\lambda\delta}{2K\rho^2}, 1 } }\nonumber\\
		&\leq  2^{4nBr}\exp\LKr{-\LPr{\frac{n\lambda\delta}{2K\rho^2 }}\min \LPr{\frac{\lambda\delta}{2K\rho^2}, 1 } },\label{eq:Pe1-c-step-1}
		\end{align}
		where the second step follows because $|\F|\leq |\C|^2\leq 2^{2nBr}$. Note that by assumption, $\delta\leq 16\rho^2/3$ and $0<\lambda<0.5$.  Therefore, 
		\[
		{\frac{\lambda\delta}{2K\rho^2}=\frac{3\lambda\delta }{16\rho^2}\leq \lambda <1.}
		\]

		Therefore, \eqref{eq:Pe1-c-step-1} can be simplified as
		\begin{align}
		\Pr\LPr{\E_1^c}&\leq  2^{4nBr}\exp\LKr{-\LPr{\delta\over 2K\rho^2}^2\lambda^2 n}.\label{eq:Pe1-c-step-2}
		\end{align}
		Conditioned on $\E_1$, Eq. \eqref{Eq:simple_2terms} becomes
		\begin{align} \label{Eq:thetaDDbar}
		\Lp{\mvec{\theta}^{t+1}}{2}{}\leq 2\lambda \Lp{\mvec{\theta}^{t}}{2}{} + 2 \Lp{\sum_{i=1}^B\Dmat_i\LPr{\mvec{x}_i-\tilde{\mvec{x}}_i}}{2}{}\Lp{ \sum_{i=1}^B\Dmat_i\bar{\mvec{\theta}_i}^{t+1}}{2}{}
		\end{align}
		
		For fixed $\bar{\mvec{\theta}}=\bar{\mvec{\theta}}'$, \eqref{eq:dev-fixed-theta} yields
		\begin{align}
		\Pr&\LPr{\Lp{\sum_{i=1}^B\Dmat_i\bar{\mvec{\theta}_i}}{2}{2}\leq 1-\epsilon }\leq  \exp\LKr{-{\frac{n\epsilon^2\delta}{4K^2\rho^2 }} }.
		\end{align}
		Using the same procedure, we can derive a counterpart to the above bound as follows 
		\begin{align}
		\Pr&\LPr{\Lp{\sum_{i=1}^B\Dmat_i\bar{\mvec{\theta}_i}}{2}{2}\geq 1+\epsilon }\leq  \exp\LKr{-{\frac{n\epsilon^2\delta}{4K^2\rho^2 }} }.\label{eq:bound-1+epsilon}
		\end{align}
		Now we consider $\Lp{\sum_{i=1}^B\Dmat_i\LPr{\mvec{x}_i-\tilde{\mvec{x}}_i}}{2}{}$ and $\Lp{ \sum_{i=1}^B\Dmat_i\bar{\mvec{\theta}_i}^{t+1}}{2}{}$ in~\eqref{Eq:thetaDDbar} separately.
		
		\begin{itemize}
			\item [a)] Given $\epsilon_1>0$, define event $\E_2$ as
		\begin{equation}
	\E_2\triangleq\left\{{1\over n}\Lp{\sum_{i=1}^B\Dmat_i\LPr{\mvec{x}_i-\tilde{\mvec{x}}_i}}{2}{2}\leq {1\over n}\Lp{\mvec{x}-\tilde{\mvec{x}}}{2}{2}+B\rho^2\epsilon_1\right\}. \label{Eq:eveen_2}
		\end{equation}
			Note that by the definition of $\delta$, we have 
			\begin{equation}
			\Lp{\mvec{x}-\tilde{\mvec{x}}}{2}{2}\le nB\delta.
			\end{equation} 
			
			From \eqref{eq:main-conv-epsilon1-simplified}, we have
			\begin{align}
			\Pr\LPr{\E_2^c}\leq  \exp\LKr{-\min \LPr{\frac{n\epsilon_1^2}{4K^2  }, \frac{n\epsilon_1}{2K}  }}.\label{eq:Pe2-c}
			\end{align}
% = K{\rho^2\over\delta}$.
			
			Conditioned on $\E_2$, we have
			\begin{align}
			\Lp{\sum_{i=1}^B\Dmat_i\LPr{\mvec{x}_i-\tilde{\mvec{x}}_i}}{2}{2}&\leq \|\xv - \tilde{\xv}\|_2^2 + nB\rho^2 \epsilon_1 \nonumber\\
			&\le nB\delta + nB\rho^2 \epsilon_1\nonumber\\
			 &= \sqrt{nB} \sqrt{\delta + \rho^2 \epsilon_1}.
			\end{align}
			
			\item [b)] Given $\epsilon_2>0$, define event $\E_3$ as
			\begin{equation}
		\E_3\triangleq \LKr{\Lp{\sum_{i=1}^B\Dmat_i\bar{\mvec{\theta}_i}}{2}{2}\leq 1+\epsilon_2 :\;\forall \;\mvec{\theta}\in\F }. \label{Eq:eveen_3}
			\end{equation}
			Combining \eqref{eq:bound-1+epsilon} with the union bound, we have
			\begin{align}
			\Pr\LPr{\E_3^c}&\leq  |\F| \exp\LKr{-{\frac{n\epsilon_2^2\delta}{4K^2\rho^2 }} } 
			\leq 2^{2nBr} \exp\LKr{-{\frac{n\epsilon_2^2\delta}{4K^2\rho^2 }} }.\label{eq:Pe3-c}
			\end{align}
		\end{itemize}
		Conditioned on $\E_1\cap\E_2\cap\E_3$, it follows from~\eqref{Eq:thetaDDbar} that 
		\begin{align}
		{1\over \sqrt{nB}} \Lp{\mvec{\theta}^{t+1}}{2}{}\leq {2\lambda \over \sqrt{nB}}\Lp{\mvec{\theta}^{t}}{2}{} + 2\sqrt{(1+\epsilon_2)(\delta+\rho^2\epsilon_1)}.\label{eq:error-final-e1-e2}
		\end{align}
	\end{itemize}
	
	Combining \eqref{eq:Pe1-c-step-2}, \eqref{eq:Pe2-c} and \eqref{eq:Pe3-c}, it follows that
	\begin{align}
	\Pr\LPr{\E_1\cap\E_2\cap\E_3}&\geq 1- \sum_{i=1}^3\Pr\LPr{\E_i^c}\nonumber\\
	&\geq  1- 2^{4nBr}\exp\LKr{-\LPr{3\delta\over 16\rho^2}^2\lambda^2 n} \nonumber\\
	& \quad-\exp\LKr{-\min \LPr{\frac{n\epsilon_1^2}{4K^2  }, \frac{n\epsilon_1}{2K}  }} - 2^{2nBr} \exp\LKr{-{\frac{n\epsilon_2^2\delta^2}{4K^2\rho^4 }} }. \label{eq:union_1}
	\end{align}
Setting $\epsilon_1 = \delta/\rho^2$ and  $\epsilon_2=1$, it follows that 
	\begin{eqnarray}
	\min \LPr{\frac{n\epsilon_1^2}{4K^2  }, \frac{n\epsilon_1}{2K}  }%&=& \frac{n\epsilon_1}{2K'}\min\left(\frac{\epsilon_1}{2K'},1\right) \nonumber \\
	&=& \frac{n\delta}{2K\rho^2}\min\left(\frac{\delta}{2K\rho^2},1\right)\\
	&=& n \left(\frac{\delta}{2K\rho^2}\right)^2,
	\end{eqnarray}
	where the last line follows because $\delta\le 2K\rho^2$ by assumption. Furthermore,
\begin{eqnarray}
\frac{n\epsilon_2^2\delta^2}{4K^2\rho^4 } &=& n \left(\frac{\delta}{2K\rho^2}\right)^2.
	\end{eqnarray}
%	Therefore, \eqref{eq:union_1} can be written as 
%	\begin{align}
%	\Pr\LPr{\E_1\cap\E_2\cap\E_3}&\geq 1- 2^{4nBr}\exp\LKr{-n\LPr{\delta\lambda\over 2K\rho^2}^2 n}  \nonumber\\
%	&\qquad- (2^{2nBr}+1) \exp\left\{-n \left(\frac{\delta}{2K\rho^2}\right)^2\right\}.
%	\end{align}
	Hence, from \eqref{eq:error-final-e1-e2}, for $t=0,1,\ldots$, 
	\begin{eqnarray}
	{1\over \sqrt{nB}} \Lp{\mvec{\theta}^{t+1}}{2}{}\leq {2\lambda \over \sqrt{nB}}\Lp{\mvec{\theta}^{t}}{2}{} + 4\sqrt{\delta},
	\end{eqnarray}
%	%
%	The main result will be
%	\begin{eqnarray}
%	{1\over \sqrt{nB}} \|\xv^{t+1} - \tilde{\xv}\|_2\leq {2\lambda \over \sqrt{nB}}\|\xv^{t} - \tilde{\xv}\|_2 + 4\sqrt{\delta},
%	\end{eqnarray}
	with a probability larger than 
	\begin{align}
	\Pr\LPr{\E_1\cap\E_2\cap\E_3}&\geq 1- 2^{4nBr}\exp\LKr{-n\LPr{\delta\lambda\over 2K\rho^2}^2 n}  \nonumber\\
	&\qquad- (2^{2nBr}+1) \exp\left\{-n \left(\frac{\delta}{2K\rho^2}\right)^2\right\}.
	\end{align}	
\end{proof}

%\clearpage
%\newpage
%-----------------------------%-----------------------------%-----------------------------%-----------------------------
\subsection{Proof of  Theorem~\ref{thm:main-PGD_noise}  \label{append:main_pgd_noise}}
\begin{proof}
Define the error vector and the normalized error vector as \eqref{eq:def-theta} and \eqref{eq:def-theta-bar}, respectively. Using the same procedure as the one used in deriving \eqref{eq:main-derivation}, and noting that unlike in \eqref{eq:main-derivation}, here, the measurements are noisy and  $\mvec{y}=\sum_{i=1}^B\Dmat_i\mvec{x}_i+\zv$, it follows that   
\begin{align}
	\sum_{i=1}^B\Lp{\mvec{\theta}_i^{k+1}}{2}{2}&\leq 2\sum_{i=1}^B \LPd{\mvec{\theta}_i^{k+1},\mvec{\theta}_i^{k}}
	-2 \LPd{\sum_{i=1}^B\Dmat_i\mvec{\theta}_i^{k}, \sum_{i=1}^B\Dmat_i\mvec{\theta}_i^{k+1}}\nonumber\\
	&\;\;-2 \LPd{\sum_{i=1}^B\Dmat_i\LPr{\mvec{x}_i-\tilde{\mvec{x}}_i}, \sum_{i=1}^B\Dmat_i\mvec{\theta}^{k+1}}\nonumber\\
	&\;\; +2\left| \sum_{i=1}^B \LPd{\Dmat_i\zv,\mvec{\theta}_i^{k+1}}\right|.
	\label{eq:main-derivation-pgd-noisy-s1}
\end{align}
Define the set of all possible normalized error vectors $\F$ as in \eqref{eq:def-set-F}. Then, given $\lambda>0$, $\epsilon_1>0$ and $\epsilon_2>0$, define events $\E_1$ $\E_2$ and $\E_3$ as \eqref{Eq:eveen_1}, \eqref{Eq:eveen_2} and \eqref{Eq:eveen_3}, respectively. Conditioned on $\E_1\cap\E_2\cap\E_3$,  using \eqref{eq:error-final-e1-e2}, \eqref{eq:main-derivation-pgd-noisy-s1} and setting $\epsilon_1$ and $\epsilon_2$ as before,  it follow that
	\begin{align}
	\Lp{\mvec{\theta}^{k+1}}{2}{}&\leq 2\lambda \Lp{\mvec{\theta}^{k}}{2}{} +4\sqrt{nB\delta} +2\left| \sum_{i=1}^B \LPd{\Dmat_i\zv,\bar{\mvec{\theta}}_i^{k+1}}\right|.
	\label{eq:main-derivation-pgd-noisy-s2}
	\end{align}
	
	To finish the proof we need to bound $\left| \sum_{i=1}^B \LPd{\Dmat_i \zv,\bar{\mvec{\theta}}_i^{k+1}}\right|$.  To achieve this goal, given  $\epsilon_z\in(0,\sqrt{\rho})$,  define event $\E_z$ as follows,
\begin{equation}
\E_z\triangleq \LKr{ \left|\sum_{i=1}^B \LPd{\Dmat_i\zv,\bar{\mvec{\theta}}_i^{k+1}}\right| \leq \sigma\sqrt{n}\epsilon_z :\;\forall \;\mvec{\theta}\in\F }. \label{Eq:eveen_z}
\end{equation}
	Then, conditioned on $\E_1\cap\E_2\cap\E_3\cap\E_z$, from \eqref{eq:main-derivation-pgd-noisy-s2}, we have
	\begin{align}
	{1\over \sqrt{nB}}\Lp{\mvec{\theta}^{k+1}}{2}{}&\leq {2\lambda\over  \sqrt{nB}} \Lp{\mvec{\theta}^{k}}{2}{} +4\sqrt{\delta} +2\epsilon_z\sigma,
	\label{eq:main-derivation-pgd-noisy-s3}
	\end{align}
	which is the desired bound. The final step is to bound $\Pr(\E_z^c)$. Note that 
	\begin{align}
	\sum_{i=1}^B \LPd{\Dmat_i\zv,\bar{\mvec{\theta}}_i^{k+1}}&=\LPd{\zv,\sum_{i=1}^B\Dmat_i\bar{\mvec{\theta}}_i^{k+1}}.
	\end{align}
	For  fixed vector $\bar{\mvec{\theta}}^{k+1}$, let $\kappav=\sum_{i=1}^B\Dmat_i\bar{\mvec{\theta}}_i^{k+1}$.  $\kappav$ is a zero-mean Gaussian vector; let $\kappa_{i_1},\kappa_{i_2}$ denote the $i_1^{th}$ and $i_2^{th}$ elements in $\kappav$, respectively. We have
	\begin{align}
	\Eox{\kappa_{i_1}\kappa_{i_2}}&=\Eox{\sum_{j_1=1}^B D_{j_1,(i_1,i_1)}\bar{\theta}_{j_1,i_1}^{k+1}\sum_{j_2=1}^BD_{j_2,(i_2,i_2)}\bar{\theta}_{j_2,i_2}^{k+1}}\nonumber\\
	&=\sum_{j_1=1}^B\sum_{j_2=1}^B\Eox{D_{j_1,(i_1,i_1)}D_{j_2,(i_2,i_2)}}\bar{\theta}_{j_1,i_1}^{k+1}\bar{\theta}_{j_2,i_2}^{k+1},\label{eq:Eox-WiWj}
	\end{align}
	where $D_{j_1,(i_1,i_1)}$ denotes the $(i_1,i_1)^{th}$ element in matrix $\Dmat_i$ and $\bar{\theta}_{j_1,i_1}^{k+1}$ denotes the $i_1^{th}$ element in the vector $\bar{\thetav}_{j_1}^{k+1}$, similar for $D_{j_2,(i_2,i_2)}$ and $\bar{\theta}_{j_2,i_2}^{k+1}$.
	
	Since $\Eox{D_{j_1,(i_1,i_1)}D_{j_2,(i_2,i_2)}}=0$, unless $j_1=j_2$ and $i_1=i_2$. Hence, for $i_1\neq i_2$, 
	\[
	\Eox{\kappa_{i_1}\kappa_{i_2}}=0.
	\]
	For $i_1=i_2=i$,  \eqref{eq:Eox-WiWj} simplifies to 
	\begin{align}
	\Eox{\kappa_{i}^2}&=\sum_{j_1=1}^B\sum_{j_2=1}^B\Eox{D_{j_1,(i,i)}D_{j_2,(i,i)}}\bar{\theta}_{j_1,i}^{k+1}\bar{\theta}_{j_2,i}^{k+1}\nonumber\\
	&=\sum_{j=1}^B\Eox{D_{j,(i,i)}^2}\LPr{\bar{\theta}_{j,i}^{k+1}}^2\nonumber\\
	&=\sum_{j=1}^B\LPr{\bar{\theta}_{j,i}^{k+1}}^2.\label{eq:var-i}
	\end{align}
	Therefore, in summary, for a fixed vector $\bar{\mvec{\theta}}^{k+1}$, $\sum_{i=1}^B\Dmat_i\bar{\mvec{\theta}}_i^{k+1}$ is a zero-mean Gaussian vector in $\mathbb{R}^n$ with  independent entries. The variance of element $i$, $\kappa_i$, is characterized in \eqref{eq:var-i}. Therefore, $\kappa_iz_i$ are independent sub-exponential random variables. Moreover, as proved in the proof of Theorem \ref{thm:main-PGD}, since by assumption $\|\mvec{\theta}^k\|_2\geq \sqrt{nB\delta}$ and $\|\mvec{\theta}^{k+1}\|_2\geq \sqrt{nB\delta}$,  for $j=1,\ldots,B$ and $i=1,\ldots,n$, we have 
	\[
	|\bar{\theta}^{k+1}_{j,i}|^2\leq {\rho^2 \over nB\delta}.%\label{eq:bound-e-bar-ij-k}
	\]
	%and, similarly,
	%\[
	%|\bar{\theta}^{k+1}_{ij}|^2\leq {\rho^2 \over nB\delta}.%\label{eq:bound-e-bar-ij-k+1}
	%\]
	Hence, 
	\begin{align}
	\Lp{\kappa_i z_i}{\psi_1}{} &\le \Lp{\kappa_i}{\psi_2}{}  \Lp{z_i}{\psi_2}{}\nonumber\\
	&={8\over 3}\sigma \LPr{\sum_{j=1}^B\LPr{\bar{\theta}_{j,i}^{k+1}}^2}^{1\over 2}\nonumber\\
	&\leq {8 \rho \sigma \over 3 \sqrt{n\delta}}.
	\end{align}
	Let $K={8\over 3}$, using Theorem \ref{thm_bound_sub_exp}, 
	\begin{align}
	&\Pr\LPr{\sum_{i=1}^B \LPd{\Dmat_i \zv,\bar{\mvec{\theta}}_i^{k+1}}\ge \sigma \sqrt{n}\epsilon_z  }= \Pr\LPr{\sum_{i=1}^n \kappa_i z_i\ge  \sigma \sqrt{n}\epsilon_z  }\nonumber\\
	& \le \exp\LKr{-\min \LPr{\frac{n\epsilon_z^2}{4{K^2\rho^2\over n\delta} n}, \frac{\sqrt{n}\epsilon_z}{2{K\rho\over \sqrt{n\delta}}} } }\nonumber\\
	& = \exp\LKr{- \LPr{\frac{n\epsilon_z\sqrt{\delta }}{2K\rho} }\min \LPr{\frac{\epsilon_z\sqrt{\delta}}{2K\rho}, 1 } }\nonumber\\
	& = \exp\LKr{- \LPr{\frac{n\epsilon_z^2\delta }{4K^2\rho^2} }} = \exp \left\{- n(\frac{3\epsilon_z }{16\rho} )^2\delta\right\},\label{eq:prob-error-term-noisy}
	\end{align}
	where the last line follows since $\epsilon_z^2\sqrt{\delta}<\rho$. Combining \eqref{eq:prob-error-term-noisy} with the union bound finishes the proof. 
\end{proof}

%-----------------------------%-----------------------------%-----------------------------%-----------------------------
\subsection{Proof of  Theorem~\ref{thm:main-GAP} \label{append:main_gap}}
\begin{proof}%[Proof of Theorem~\ref{thm:main-GAP} (Theorem 3 in the main paper)]
	Recall that $\Dmat_i= \diag(D_{i1},\dots,D_{in})$, $\Hmat = [\Dmat_1, \dots, \Dmat_B]$, and $\Rmat = \Hmat\Hmat\ts = \diag(R_1,\dots, R_n)$,
	where
	\[
	R_{j}\triangleq \sum_{i=1}^BD_{ij}^2.
	\]
	Assume that ${1\over nB}\Lp{\tilde{\mvec{x}}-{\mvec{x}}^{t}}{2}{2}\geq  \delta$. We need to prove that 
	%%\vspace{-1mm}
	\[
	{{1\over \sqrt{nB}} \|{\mvec{x}}^{t+1}-\tilde{\mvec{x}}\|_2\leq {2\lambda \over \sqrt{nB}}\Lp{{\mvec{x}}^{t}-\tilde{\mvec{x}}}{2}{} + 4\sqrt{\delta}},
	\]
	holds with high probability. 

	Using   a derivation similar to \eqref{eq:main-derivation} in Section~\ref{append:main_pgd} and noting that in GAP
	\[
	\mvec{s}_i^{t+1}=\mvec{x}_i^{t}+{B}\Dmat_i\Rmat^{-1}\mvec{e}^t,
	\]
	it follows that
	\begin{align}
	&\sum_{i=1}^B\Lp{\tilde{\mvec{x}}_i-{\mvec{x}}_i^{t+1}}{2}{2}
	\leq 2\sum_{i=1}^B \LPd{\tilde{\mvec{x}}_i-\mvec{s}_i^{t+1},\tilde{\mvec{x}}_i-{\mvec{x}}_i^{t+1}}\nonumber\\
	&= 2\sum_{i=1}^B \LPd{\tilde{\mvec{x}}_i-{\mvec{x}}^{t}_i- {B}\Dmat_i\Rmat\inv\mvec{e}^t,\tilde{\mvec{x}}_i-{\mvec{x}}_i^{t+1}}\nonumber\\
	&= 2\sum_{i=1}^B \LPd{\tilde{\mvec{x}}_i-{\mvec{x}}^{t}_i,\tilde{\mvec{x}}_i-{\mvec{x}}_i^{t+1}} \nonumber\\
	&\quad-2 {B}\sum_{i=1}^B \LPd{\Dmat_i\Rmat\inv\LPr{\mvec{y}-\sum_{j=1}^B\Dmat_j\mvec{x}_j^t},\tilde{\mvec{x}}_i-{\mvec{x}}_i^{t+1}}\nonumber\\
	&=2\sum_{i=1}^B \LPd{\tilde{\mvec{x}}_i-{\mvec{x}}^{t}_i,\tilde{\mvec{x}}_i-{\mvec{x}}_i^{t+1}}\nonumber\\
	&\quad
	-2{B} \sum_{i=1}^B \LPd{\sum_{j=1}^B\Dmat_j (\xv_j-\mvec{x}_j^t),\Rmat\inv \Dmat_i(\tilde{\mvec{x}}_i-{\mvec{x}}_i^{t+1})}\nonumber\\
%	\end{align}
	%&=2\sum_{i=1}^B \LPd{\tilde{\mvec{x}}_i-{\mvec{x}}^{t}_i,\tilde{\mvec{x}}_i-{\mvec{x}}_i^{t+1}}
	%-2 \sum_{i=1}^B \LPd{\Dmat_i\LPr{\sum_{j=1}^B\Dmat_j(\xv_j-\mvec{x}_j^t)},\tilde{\mvec{x}}_i-{\mvec{x}}_i^{t+1}}\nonumber\\
	%&= 2\sum_{i=1}^B \LPd{\tilde{\mvec{x}}_i-{\mvec{x}}^{t}_i,\tilde{\mvec{x}}_i-{\mvec{x}}_i^{t+1}}
	%-2 \sum_{i=1}^B \LPd{\sum_{j=1}^B\Dmat_j\LPr{\mvec{x}_j-\mvec{x}_j^t},\Dmat_i\LPr{\tilde{\mvec{x}}_i-{\mvec{x}}_i^{t+1}}}\nonumber\\
	%&= 2\sum_{i=1}^B \LPd{\tilde{\mvec{x}}_i-{\mvec{x}}^{t}_i,\tilde{\mvec{x}}_i-{\mvec{x}}_i^{t+1}}
	%-2  \LPd{\sum_{j=1}^B\Dmat_j\LPr{\mvec{x}_j-\mvec{x}_j^t},\sum_{i=1}^B\Dmat_i\LPr{\tilde{\mvec{x}}_i-{\mvec{x}}_i^{t+1}}}\nonumber\\
%\begin{align}
	&= 2\sum_{i=1}^B \LPd{\tilde{\mvec{x}}_i-{\mvec{x}}^{t}_i,\tilde{\mvec{x}}_i-{\mvec{x}}_i^{t+1}}\nonumber\\
	&\quad
	-2{B} \LPd{\sum_{i=1}^B\Dmat_i\LPr{\mvec{x}_i-\tilde{\mvec{x}}_i+\tilde{\mvec{x}}_i-\mvec{x}_i^t}, \Rmat\inv\sum_{i=1}^B\Dmat_i\LPr{\tilde{\mvec{x}}_i-{\mvec{x}}_i^{t+1}}}\nonumber\\
	&= 2\sum_{i=1}^B \LPd{\tilde{\mvec{x}}_i-{\mvec{x}}^{t}_i,\tilde{\mvec{x}}_i-{\mvec{x}}_i^{t+1}}\nonumber\\
	&\quad
	-2 {B}\LPd{\sum_{i=1}^B\Dmat_i\LPr{\tilde{\mvec{x}}_i-\mvec{x}_i^t}, \Rmat\inv\sum_{i=1}^B\Dmat_i\LPr{\tilde{\mvec{x}}_i-{\mvec{x}}_i^{t+1}}}\nonumber\\
	&\quad-2 {B}\LPd{\sum_{i=1}^B\Dmat_i\LPr{\mvec{x}_i-\tilde{\mvec{x}}_i}, \Rmat\inv\sum_{i=1}^B\Dmat_i\LPr{\tilde{\mvec{x}}_i-{\mvec{x}}_i^{t+1}}}.\label{eq:gap_err_3terms}
	\end{align} 
	
	Defining  ${\mvec{\theta}}^t$ and $\bar{\mvec{\theta}}^t$ as in \eqref{eq:def-theta} and \eqref{eq:def-theta-bar}, respectively,  the first two terms in~\eqref{eq:gap_err_3terms} can be written as 
	\begin{align}
	&2\sum_{i=1}^B \LPd{\tilde{\mvec{x}}_i-{\mvec{x}}^{t}_i,\tilde{\mvec{x}}_i-{\mvec{x}}_i^{t+1}}\nonumber\\
	&
	\quad -2 {B} \LPd{\sum_{i=1}^B\Dmat_i\LPr{\tilde{\mvec{x}}_i-\mvec{x}_i^t}, \Rmat\inv\sum_{i=1}^B\Dmat_i\LPr{\tilde{\mvec{x}}_i-{\mvec{x}}_i^{t+1}}} \nonumber\\
	&=2 \Lp{\mvec{\theta}^{t}}{2}{} \Lp{\mvec{\theta}^{t+1}}{2}{} \left(\sum_{i=1}^B \LPd{\bar{\mvec{\theta}}^t_i,\bar{\mvec{\theta}}_i^{t+1}} \right.\nonumber\\
	&
	\qquad
	-  \left.{B}\LPd{\sum_{i=1}^B\Dmat_i\bar{\mvec{\theta}}^t_i, \Rmat\inv\sum_{i=1}^B\Dmat_i\bar{\mvec{\theta}}^{t+1}_i}\right). \label{eq:first2term_err_gap}
	\end{align}
	Similarly, the third term in \eqref{eq:gap_err_3terms} can be written as 
	\begin{align}
	&\LPd{\sum_{i=1}^B\Dmat_i\LPr{\mvec{x}_i-\tilde{\mvec{x}}_i}, \Rmat\inv\sum_{i=1}^B\Dmat_i\LPr{\tilde{\mvec{x}}_i-{\mvec{x}}_i^{t+1}}}\nonumber\\
	&=\|\thetav ^{t+1}\|_2 \LPd{\sum_{i=1}^B\Dmat_i\LPr{\mvec{x}_i-\tilde{\mvec{x}}_i}, \Rmat\inv\sum_{i=1}^B\Dmat_i\bar{\thetav}_i^{t+1}}. \label{eq:gap_term3}
	\end{align}
	Hence, in summary, \eqref{eq:gap_err_3terms} can be written as
	\begin{align}
	&\Lp{\mvec{\theta}^{t+1}}{2}{}\leq 2B \Labs{\LPd{\sum_{i=1}^B\Dmat_i\LPr{\mvec{x}_i-\tilde{\mvec{x}}_i}, \Rmat\inv\sum_{i=1}^B\Dmat_i\bar{\thetav}_i^{t+1}}} \nonumber\\
	&+2 \Lp{\mvec{\theta}^{t}}{2}{} (\sum_{i=1}^B \LPd{\bar{\mvec{\theta}}^t_i,\bar{\mvec{\theta}}_i^{t+1}}
	-  {B}\LPd{\sum_{i=1}^B\Dmat_i\bar{\mvec{\theta}}^t_i, \Rmat\inv\sum_{i=1}^B\Dmat_i\bar{\mvec{\theta}}^{t+1}_i}).\label{eq:GAP-main-simplified}
	\end{align} 
	
	Note that
	\begin{align}
	&\sum_{i=1}^B \LPd{\bar{\mvec{\theta}}^t_i,\bar{\mvec{\theta}}_i^{t+1}}
	- {B}\LPd{\sum_{i=1}^B\Dmat_i\bar{\mvec{\theta}}^t_i, \Rmat\inv\sum_{i=1}^B\Dmat_i\bar{\mvec{\theta}}^{t+1}_i}\nonumber\\
	&=\sum_{j=1}^n\LPr{\sum_{i=1}^B \bar{\theta}^t_{ij}\bar{\theta}_{ij}^{t+1}
		-\frac{ {B}}{R_{j}}\LPr{\sum_{i_1=1}^B  D_{i_1j}\bar{\theta}^t_{i_1j}} \LPr{ \sum_{i_2=1}^B D_{i_2j}\bar{\theta}^{t+1}_{i_2j}}}\label{eq:main-GAP-2-terms}
	\end{align}
	For $j=1,\ldots,n$, define random variable $P_j$ as follows
	\begin{align}
	P_j\;\triangleq\;\frac{ {B}}{R_{j}}\LPr{\sum_{i_1=1}^B  D_{i_1j}\bar{\theta}^t_{i_1j}} \LPr{ \sum_{i_2=1}^B D_{i_2j}\bar{\theta}^{t+1}_{i_2j}}.
	\end{align}
	We first show that $P_j$ is a bounded random variable. Applying   the Cauchy-Schwarz inequality to both terms in $P_j$, we have 
	\begin{align}
	|P_j|&={{B}\over R_j}\Labs{\sum_{i_1=1}^B  D_{i_1j}\bar{\theta}^t_{i_1j}} \Labs{ \sum_{i_2=1}^B D_{i_2j}\bar{\theta}^{t+1}_{i_2j}}\nonumber\\
	&\leq {{B}\over R_j}{ \sum_{i=1}^BD_{ij}^2 \sqrt{\sum_{i=1}^B  (\bar{\theta}^t_{ij})^2}\sqrt{\sum_{i=1}^B  (\bar{\theta}^{t+1}_{ij})^2}}\nonumber\\
	&=  {B} \sqrt{\sum_{i=1}^B  (\bar{\theta}^t_{ij})^2 \sum_{i=1}^B  (\bar{\theta}^{t+1}_{ij})^2}.\label{eq:bound-U-j}
	\end{align}
	Next, we find the expected value of $P_j$. Note that
	\begin{align}
	&\Eox{P_j}=\Eox{\frac{ {B}}{R_{j}}\LPr{\sum_{i_1=1}^B  D_{i_1j}\bar{\theta}^t_{i_1j}} \LPr{ \sum_{i_2=1}^B D_{i_2j}\bar{\theta}^{t+1}_{i_2j}}}\nonumber\\
	&=
	B\sum_{i_1=1}^B\sum_{i_2=1}^B\Eox{ D_{i_1j}D_{i_2j} \over  R_j}\bar{\theta}^t_{i_1j}  \bar{\theta}^{t+1}_{i_2j}.\label{eq:Dij-terms-all}
	\end{align}
	But, 
	\begin{align}
	\sum_{i=1}^B\Eox{ D_{ij}^2 \over  R_j}=\Eox{ \sum_{i=1}^B D_{ij}^2 \over  R_j}={R_j\over R_j}=1.\label{eq:sum-Bj-1}
	\end{align}
	Moreover, by symmetry of the distributions, 
	\begin{align}
	\Eox{ D_{1j}^2 \over  R_j}=\ldots=\Eox{ D_{Bj}^2 \over  R_j}.\label{eq:Bj-equal}
	\end{align}
	Combing \eqref{eq:sum-Bj-1} and \eqref{eq:Bj-equal}, it follows that
	\begin{align}
	\Eox{ D_{1j}^2 \over  R_j}=\ldots=\Eox{ D_{Bj}^2 \over  R_j}={1\over B}. \label{eq:Eox-1-over-B}
	\end{align}
	On the other hand, for $i_1\neq n_2$,  since $D_{i_1j}$ and $D_{i_2j}$ have symmetric distributions around zero, and are independent, we have 
	\begin{align}
	\Eox{ D_{i_1j}D_{i_2j} \over  R_j}=0.\label{eq:Eox-cross-terms-zero}
	\end{align}
	Hence, inserting \eqref{eq:Eox-1-over-B} and \eqref{eq:Eox-cross-terms-zero} in \eqref{eq:Dij-terms-all}, we have
	\begin{align}
	\Eox{P_j}&=\Eox{\frac{ {B}}{R_{j}}\LPr{\sum_{i_1=1}^B  D_{i_1j}\bar{\theta}^t_{i_1j}} \LPr{ \sum_{i_2=1}^B D_{i_2j}\bar{\theta}^{t+1}_{i_2j}}}\nonumber\\
	&=
	\sum_{i=1}^B\bar{\theta}^t_{ij}  \bar{\theta}^{t+1}_{ij}.
	\end{align}
	These results show that \eqref{eq:main-GAP-2-terms} includes the sum of $n$ independent bounded random variables. Therefore, using the Hoeffding's inequality, for fixed $\bar{\theta}^t$ and $\bar{\theta}^{t+1}$ , for any $\lambda>0$, we have
	\begin{align} 
	&\Pr\LPr{\sum_{j=1}^n\Eox{P_j}-\sum_{j=1}^nP_j\geq \lambda}\nonumber\\
	&\leq \exp\LKr{-{2\lambda^2\over  4{B}^2 \sum_{j=1}^n\sum_{i=1}^B  (\bar{\theta}^t_{ij})^2 \sum_{i=1}^B  (\bar{\theta}^{t+1}_{ij})^2}}.\label{eq:Hoeffding-main}
	\end{align}
	Note that  by assumption 
	\[
	{1\over nB}\sum_{i=1}^B\Lp{\tilde{\mvec{x}}_i-{\mvec{x}}_i^{t}}{2}{2}\geq  \delta.
	\]
	Therefore, as argued before, for $i=1,\ldots,B$ and $j=1,\ldots,n$ in the proof of Theorem~\ref{thm:main-PGD},
	\begin{align}
	|\bar{\theta}^t_{ij}|^2&={(\tilde{{x}}_{ij}-{x}^{t}_{ij})^2\over \Lp{\mvec{\theta}^t}{2}{2}}\leq {\rho^2 \over nB\delta}\label{eq:bound-e-bar-ij-k-new-proof}
	\end{align}
	and, similarly,
	\begin{align}
	|\bar{\theta}^{t+1}_{ij}|^2\leq {\rho^2 \over nB\delta}.\label{eq:bound-e-bar-ij-k+1-new-proof}
	\end{align}
	Using the bounds in \eqref{eq:bound-e-bar-ij-k+1-new-proof} and \eqref{eq:bound-e-bar-ij-k+1-new-proof}, it follows that 
	\begin{align}
	&{B}^2 \sum_{j=1}^n\sum_{i=1}^B  (\bar{\theta}^t_{ij})^2 \sum_{i=1}^B  (\bar{\theta}^{t+1}_{ij})^2\leq {B}^2 \sum_{j=1}^n\sum_{i=1}^B \LPr{{\rho^2 \over nB\delta}}^2 \nonumber\\
	&={B}^2 nB\LPr{{\rho^2 \over nB\delta}}^2={B\rho^4\over \delta^2n}.
	\end{align}
	Hence, from \eqref{eq:Hoeffding-main}, for fixed $\bar{\theta}^t$ and $\bar{\theta}^{t+1}$ , for any $\lambda>0$,  we have
	\begin{eqnarray} 
	\Pr\LPr{\sum_{j=1}^n\Eox{P_j}-\sum_{j=1}^nP_j\geq \lambda}&\leq &\exp\LKr{{-{\lambda^2\delta^2 n\over 2B\rho^4}}}.\label{eq:Hoeffding-result}
	\end{eqnarray}
	Given $\lambda>0$, define event ${\E}_1$ as follows
	\begin{align}
	{\E}_1\triangleq \left\{\sum_{i=1}^B \LPd{\bar{\mvec{\theta}}^t_i,\bar{\mvec{\theta}}_i^{t+1}}
	- {B}\LPd{\sum_{i=1}^B\Dmat_i\bar{\mvec{\theta}}^t_i, \Rmat\inv\sum_{i=1}^B\Dmat_i\bar{\mvec{\theta}}^{t+1}_i} \leq \lambda :\;\forall \;(\bar{\mvec{\theta}}, \bar{\mvec{\theta}}')\in\F^2 \right\},\label{eq:proof-gap_event_1}
	\end{align}
	where  the set of normalized error vectors $\F$ is defined before in \eqref{eq:def-set-F}.
	Then, by the union bound, we have
	\begin{align}
	\Pr\LPr{{\E}_1^c}&\leq |\F|^2\exp\LKr{-{\lambda^2\delta^2 n\over 2B\rho^4}} \nonumber\\
	&\leq  2^{4nBr}\exp\LKr{-{\lambda^2\delta^2 n\over 2B\rho^4}}.\label{eq:gap_Pe1-c-step-1}
	\end{align}

	%Next we bound the third term in  \eqref{eq:gap_err_3terms}, which is $ \LPd{\sum_{i=1}^B\Dmat_i\LPr{\mvec{x}_i-\tilde{\mvec{x}}_i}, \Rmat\inv\sum_{i=1}^B\Dmat_i\LPr{\tilde{\mvec{x}}_i-{\mvec{x}}_i^{t+1}}}$. First, dividing and multiplying the term with $\|\thetav ^{t+1}\|_2$, it can be written as
	%\begin{eqnarray}
	%\LPd{\sum_{i=1}^B\Dmat_i\LPr{\mvec{x}_i-\tilde{\mvec{x}}_i}, \Rmat\inv\sum_{i=1}^B\Dmat_i\LPr{\tilde{\mvec{x}}_i-{\mvec{x}}_i^{t+1}}}
	%&=&\|\thetav ^{t+1}\|_2 \LPd{\sum_{i=1}^B\Dmat_i\LPr{\mvec{x}_i-\tilde{\mvec{x}}_i}, \Rmat\inv\sum_{i=1}^B\Dmat_i\bar{\thetav}_i^{t+1}}.%&\le& \textcolor{green}{\frac{1}{\gamma_{\rm min}}\left\|\sum_{i=1}^B\Dmat_i\LPr{\mvec{x}_i-\tilde{\mvec{x}}_i}\right\|_2 \left\|\sum_{i=1}^B\Dmat_i{\thetav_i^{t+1}}\right\|_2}. \label{eq:gap_term3}
	%\end{eqnarray}
	We next bound the last term in \eqref{eq:GAP-main-simplified}.  Note that 
	\begin{align}
	&\LPd{\sum_{i=1}^B\Dmat_i\LPr{\mvec{x}_i-\tilde{\mvec{x}}_i}, \Rmat\inv\sum_{i=1}^B\Dmat_i\bar{\thetav}_i^{t+1}}\nonumber\\
	&=
	\sum_{j=1}^n{1\over R_j}{\sum_{i=1}^BD_{ij}\LPr{{x}_{ij}-\tilde{{x}}_{ij}}}{ \sum_{i=1}^BD_{ij}\bar{\theta}_{ij}^{t+1}}\nonumber\\
	&=
	\sum_{j=1}^n\sum_{i_1=1}^B\sum_{i_2 =1}^B{D_{i_1j}D_{i_2j}\over R_j}\LPr{{x}_{i_1j}-\tilde{x}_{i_1j}} \bar{\theta}_{i_2j}^{t+1}.
	%&\le& \textcolor{green}{\frac{1}{\gamma_{\rm min}}\left\|\sum_{i=1}^B\Dmat_i\LPr{\mvec{x}_i-\tilde{\mvec{x}}_i}\right\|_2 \left\|\sum_{i=1}^B\Dmat_i{\thetav_i^{t+1}}\right\|_2}. \label{eq:gap_term3}
	\end{align}
	For a fixed $\bar{\thetav}^{t+1}$, and $j=1,\ldots,n$, define random variable $Q_j$ as
	\begin{eqnarray}
	Q_j=\sum_{i_1=1}^B\sum_{i_2 =1}^B{D_{i_1j}D_{i_2j}\over R_j}\LPr{{x}_{i_1j}-\tilde{x}_{i_1j}} \bar{\theta}_{i_2j}^{t+1}.
	\end{eqnarray}
	Note that 
	\begin{eqnarray}
	\Eox{Q_j}&=&\sum_{i_1=1}^B\sum_{i_2 =1}^B\Eox{D_{i_1j}D_{i_2j}\over R_j}\LPr{{x}_{i_1j}-\tilde{x}_{i_1j}} \bar{\theta}_{i_2j}^{t+1}\nonumber\\
	&=&{1\over B}\sum_{i=1}^B\LPr{{x}_{i_j}-\tilde{x}_{ij}} \bar{\theta}_{ij}^{t+1}.
	\end{eqnarray}
	where the last line follows from \eqref{eq:Eox-1-over-B} and \eqref{eq:Eox-cross-terms-zero}.
	Therefore,
	\begin{eqnarray}
	\sum_{j=1}^n\Eox{Q_j}&=&{1\over B}\sum_{j=1}^n\sum_{i=1}^B\LPr{{x}_{i_j}-\tilde{x}_{ij}} \bar{\theta}_{ij}^{t+1}\nonumber\\
	&=&{1\over B}\LPd{\mvec{x}-\tilde{\mvec{x}}, \bar{\thetav}^{t+1}}.
	\end{eqnarray}
	Hence, by the Cauchy-Schwartz inequality, 
	\begin{eqnarray}
	\Labs{\sum_{j=1}^n\Eox{Q_j}}&\leq &{1\over B}\Lp{\mvec{x}-\tilde{\mvec{x}}}{2}{} \|\bar{\thetav}^{t+1}\|_2\nonumber\\
	&\stackrel{ \|\bar{\thetav}^{t+1}\|_2 =1}{\leq}& {\sqrt{nB\delta}\over B}=\sqrt{n\delta \over B}.\label{eq:bound-sum-Vj-Eox}
	\end{eqnarray}
	Moreover, applying the Cauchy-Schwartz inequality twice, it follows that
	\begin{eqnarray}
	\Labs{Q_j}&=&\Labs{\sum_{i_1=1}^B\sum_{i_2 =1}^B{D_{i_1j}D_{i_2j}\over R_j}\LPr{{x}_{i_1j}-\tilde{x}_{i_1j}} \bar{\theta}_{i_2j}^{t+1}}\nonumber\\
	&\leq &{\sum_{i=1}^BD_{i,j}^2\over R_j} \sqrt{\sum_{i=1}^B\LPr{{x}_{ij}-\tilde{x}_{ij}}^2\sum_{i=1}^B\LPr{ \bar{\theta}_{ij}^{t+1}}^2}\nonumber\\
	&= & \sqrt{\sum_{i=1}^B\LPr{{x}_{ij}-\tilde{x}_{ij}}^2\sum_{i=1}^B\LPr{ \bar{\theta}_{ij}^{t+1}}^2}.
	\end{eqnarray}
	Hence, $Q_j$, $j=1,\ldots,n$, are bounded independent random variables. Therefore, by the Hoeffding's inequality, it follows that, for any $\lambda'>0$,
	\begin{align}
	&\Pr\left\{\LPd{\sum_{i=1}^B\Dmat_i\LPr{\mvec{x}_i-\tilde{\mvec{x}}_i}, \Rmat\inv\sum_{i=1}^B\Dmat_i\bar{\thetav}_i^{t+1}}
	% \right.\nonumber\\
	%&\qquad\left.
	-\sum_{j=1}^n{1\over B}\sum_{i=1}^B\LPr{{x}_{i_j}-\tilde{x}_{ij}} \bar{\theta}_{ij}^{t+1}\geq \lambda'\right\}\nonumber\\
	&\leq  \exp\left\{-{2\lambda'^2\over  4\sum_{j=1}^n\sum_{i=1}^B\LPr{{x}_{ij}-\tilde{x}_{ij}}^2\sum_{i=1}^B\LPr{ \bar{\theta}_{ij}^{t+1}}^2}\right\}\label{eq:Hoeffding-second-use}
	\end{align}
	But, from \eqref{eq:bound-e-bar-ij-k+1-new-proof}, it follows that 
	\begin{align}
	&\sum_{j=1}^n\sum_{i=1}^B\LPr{{x}_{ij}-\tilde{x}_{ij}}^2\sum_{i=1}^B\LPr{ \bar{\theta}_{ij}^{t+1}}^2
	\leq  \sum_{j=1}^n\sum_{i=1}^B\LPr{{x}_{ij}-\tilde{x}_{ij}}^2\sum_{i=1}^B  {\rho^2 \over nB\delta}\nonumber\\
	&\qquad= {\rho^2 \over n\delta}\Lp{\mvec{x}-\tilde{\mvec{x}}}{2}{2}\leq  {\rho^2  B}.
	\end{align}
	Letting 
	\[
	%$\textcolor{green}{\lambda'=\sqrt{n\delta B}}, 
	{\lambda'=\sqrt{n\delta \over B}},
	\]
	we have from \eqref{eq:Hoeffding-second-use} that
	\begin{align}
	\Pr&\left\{\LPd{\sum_{i=1}^B\Dmat_i\LPr{\mvec{x}_i-\tilde{\mvec{x}}_i}, \Rmat\inv\sum_{i=1}^B\Dmat_i\bar{\thetav}_i^{t+1}} -\sum_{j=1}^n{1\over B}\sum_{i=1}^B\LPr{{x}_{ij}-\tilde{x}_{ij}} \bar{\theta}_{ij}^{t+1}\geq \sqrt{n\delta \over B}\right\}\nonumber\\
	&\leq  \exp\LKr{-{n\delta\over 2\rho^2 B^2}}.\label{eq:final-bd-E2-single-term}
	\end{align}
	Define event $\E_2$ as 
	\begin{align}
	\E_2=&\left\{\LPd{\sum_{i=1}^B\Dmat_i\LPr{\mvec{x}_i-\tilde{\mvec{x}}_i}, \Rmat\inv\sum_{i=1}^B\Dmat_i\bar{\thetav}_i^{t+1}}
\leq 2\sqrt{n\delta \over B}:\; \forall\bar{\thetav}\in\F \right\}.
	\end{align}
	Then,  using \eqref{eq:final-bd-E2-single-term}, \eqref{eq:bound-sum-Vj-Eox} and the  union bound, we have
	\begin{align}
	\Pr(\E_2^c) &\leq |\F|\exp\LKr{-{n\delta\over 2\rho^2 B^2}}\nonumber\\
	&\leq 2^{2nBr} \exp\LKr{-{n\delta\over 2\rho^2 B^2}}.
	\end{align}
	Finally, condition on $\E_1\cap\E_2$, it follows from \eqref{eq:first2term_err_gap} that
	\begin{align}
	\Lp{\mvec{\theta}^{t+1}}{2}{}&\leq 2 \lambda \Lp{\mvec{\theta}^{t}}{2}{}	+4B\sqrt{n\delta \over B},
	\end{align}
	or 
	\begin{align}
	{1\over \sqrt{nB}}\Lp{\xv^{t+1}-\tilde{\xv}}{2}{}&\leq {2 \lambda \over  \sqrt{nB}}\Lp{\xv^{t}-\tilde{\xv}}{2}{}	+4\sqrt{\delta},
	\end{align}
	which is the desired result.
	And the probability is
	\begin{equation}
	1-2^{2nBr} \exp\LKr{-{n\delta\over 2\rho^2 B^2}} - 2^{4nBr}\exp\LKr{-{\lambda^2\delta^2 n\over 2B\rho^4}}.
	\end{equation}
\end{proof}
%

%~~~~~~~~~~~~~~~~~~~~~~~~~~~~~~%~~~~~~~~~~~~~~~~~~~~~~~~~~~~~~%~~~~~~~~~~~~~~~~~~~~~~~~~~~~~~
%~~~~~~~~~~~~~~~~~~~~~~~~~~~~~~%~~~~~~~~~~~~~~~~~~~~~~~~~~~~~~%~~~~~~~~~~~~~~~~~~~~~~~~~~~~~~

\section{Conclusions \label{sec:end}}
We have studied the problem of snapshot compressive sensing that rises in many modern compressive imaging applications. In such systems multiple signal frames are combined with each other, such that elements with the same (spatial, physical, etc.) locations are linearly combined together to form the measured signal frame. Therefore, the measured frame has the same dimensions of a single signal frame. 
A compression-based framework has been developed to theoretically analyze the snapshot compressive sensing systems.
We have proposed two efficient recovery algorithms that  reconstruct a high dimensional signal from its snapshot measurements. The proposed algorithms are compression-based recovery schemes that are  theoretically proved to converge to the desired solution.   Our simulation results  for video snapshot CS   show that the proposed methods achieve  state-of-the art performance. 

The proposed algorithms can be utilized in various snapshot compressive sensing applications, including videos, hyperspectral images, 3D scene compressive imaging~\cite{Sun16OE,Sun17OE} and so on, thus filled the gap between existing  compressive sensing theories and practical applications. 
We expect our proposed compression-based framework to be applied to and to inspire compressed sensing systems that capture ultrafast~\cite{Gao14_Nature} and ultrahigh dimensional~\cite{Huang91OCT,Kirmani14FPI,Tsai15OL} data, and thus enable more advanced exploration of information in the nature.   

%--------------------------------------Appendices--------------------------------------
\setcounter{equation}{0}
\renewcommand{\theequation}{\thesection.\arabic{equation}}

\appendices
\section{Details of GAP-NLS}\label{App:A}
One drawback of the JPEG-based MPEG compression is that,  since it relies on discrete Cosine transformation (DCT) of individual local patches, it does not exploit  {\em nonlocal} similarities.
In this section, we detail the proposed GAP-NLS algorithm. More specifically, we describe the  encoder and decoder of the proposed compression code (NLS) that exploits nonlocal structures  in videos.
Fig.~\ref{fig:NLS} depicts our proposed NLS encoder/decoder.

\begin{figure}[htbp!]
	%\vspace{-3mm}
	\centering
	\includegraphics[width=11cm]{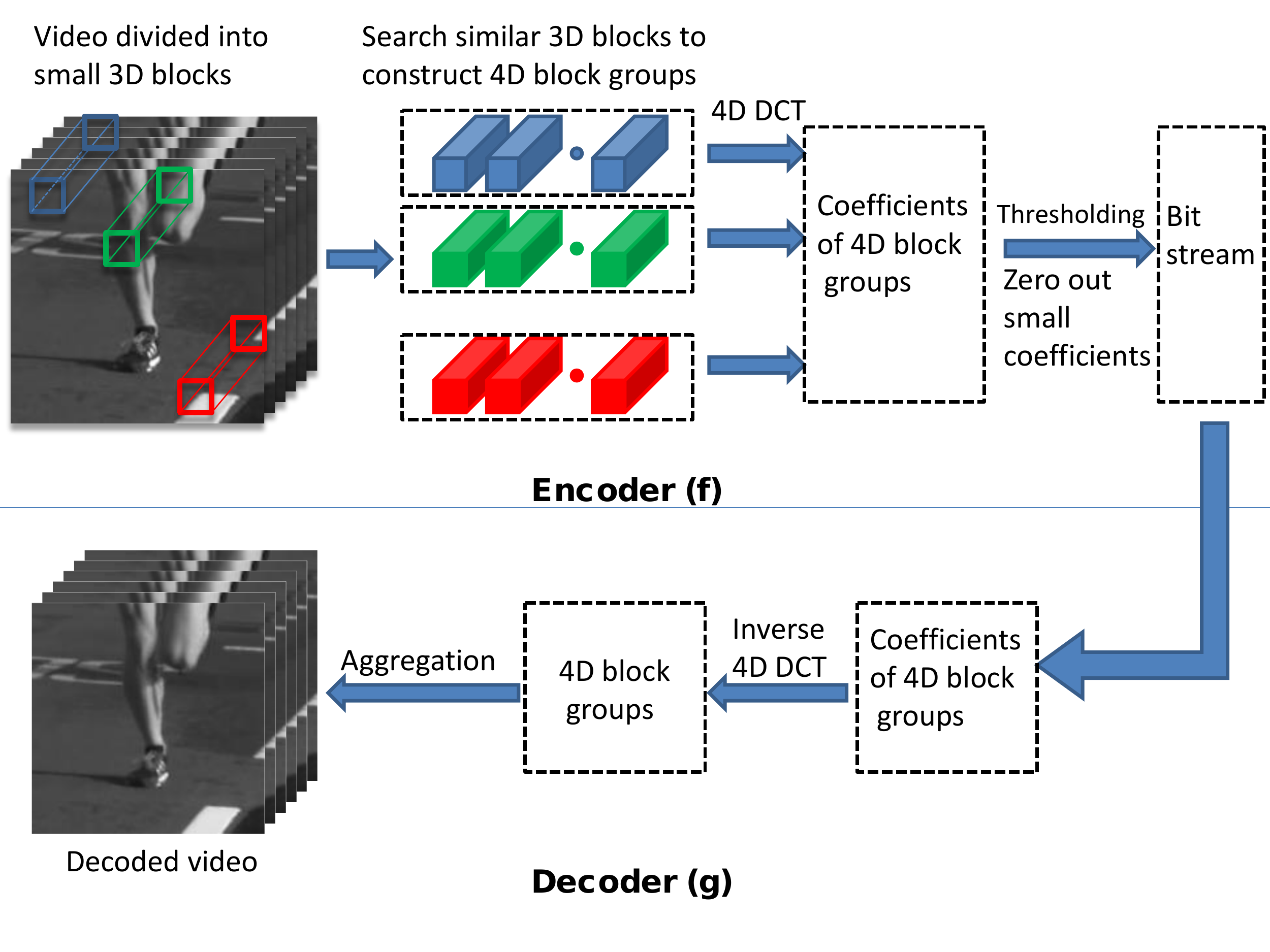}
	%\vspace{-8mm}
	\caption{Encoder (upper part) and decoder (lower part) by exploiting the nonlocal similarity of videos. }
	\label{fig:NLS}
	%\vspace{-2mm}
\end{figure}

The encoding operation includes performing the following steps:
\begin{itemize}
	\item [i)] Divide the 3D video $\Xmat \in{\mathbb R}^{n_x \times n_y \times B}$  into small overlapping blocks  ${\cal X}_q \in {\mathbb R}^{p_x \times p_y \times B}$, $q=1,\dots, Q$, where $Q$ denotes the number of such small 3D blocks.
	\item [ii)] For each {\em reference} block, ${\cal X}_q$,  find its ``similar'' blocks between the remaining  $Q-1$ blocks. The similarity between two blocks is measured in terms of  their $\ell_2$-norm distance.  For ${\cal X}_q$, the encoder picks its  $G$ most similar blocks, thus forms a 4-way tensor $\tilde{\cal X}_q \in {\mathbb R}^{p_x \times p_y \times B \times G}$. These similar blocks are sorted based on their $\ell_2$-norm distances, from the smallest to the largest and thus the reference block is always the first one.
	\item [iii)] Apply 4D DCT (discrete Cosine transformation)  to each 4D block group, and derive  the 4D coefficients for each block group.  For each group of blocks, since the  blocks are very similar, the coefficients will be sparse in the transform domain. 
	\item [iv)] Set the coefficients with small amplitudes  to zero. More specifically,  the encoder only keeps  $p_x p_y B$ (out of $p_xp_y b G$) coefficients with largest amplitudes for each block group. 	In total, there will be $p_x p_y B Q$ non-zero coefficients for the entire video since there are $Q$ 3D blocks. 
	\item [v)] {Encode the remaining $p_x p_y B Q$  non-zero coefficients and their corresponding locations and their group indices.}
\end{itemize}

Note that there are significant redundancies in this encoder since each pixel has been encoded many times, which will help the decoder to achieve videos with higher quality.

Correspondingly, the decoder includes the following steps.
\begin{itemize}
	\item [i)] {Decode the non-zero coefficients and their  corresponding locations and group indices.}
	\item [ii)] Fill in the zero coefficients to get the coefficients of 4D block groups.
	\item [iii)] Perform inverse 4D DCT to the coefficients of each 4D block group.
	\item [iv)] {Put the blocks back to the original locations and aggregate these blocks (taking average value of each pixel) to achieve the decoded video.}
\end{itemize}

\bibliographystyle{IEEEtran}
\bibliography{myrefs,reference_sideinfor,reference_ECCV}

\end{document}